\documentclass[a4paper,UKenglish,cleveref,autoref,thm-restate]{lipics-v2021}

\graphicspath{{./figures/}}

\usepackage{amsfonts}
\usepackage{mathtools}
\usepackage{algorithm,algorithmic}
\usepackage{nicefrac}
\usepackage{complexity}
\usepackage{xspace}

\usepackage{cleveref}
\crefname{claim}{Claim}{Claims}

\hideLIPIcs

\newcommand{\newterm}[1]{\textbf{\textit{#1}}}
\newcommand{\phaseref}[1]{\textsf{#1}}

\newcommand{\north}{\textsc{n}}
\newcommand{\east}{\textsc{e}}
\newcommand{\south}{\textsc{s}}
\newcommand{\west}{\textsc{w}}

\newdimen\configfigureheight%
\newdimen\tmpconfigfigureheight%
\newcommand{\transformable}[1]{%
    \settoheight{\tmpconfigfigureheight}{#1}%
    \global\configfigureheight=\tmpconfigfigureheight%
    #1%
}
\newcommand{\transforms}{%
    \makebox(24pt,\configfigureheight){\includegraphics[page=1]{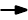}}%
}
\newcommand{\ntransforms}{%
    \makebox(24pt,\configfigureheight){\includegraphics[page=2]{icons-transform-arrow}}%
}
\newcommand{\ttransforms}{%
    \makebox(24pt,\configfigureheight){\includegraphics[page=3]{icons-transform-arrow}}%
}

\title{Sliding Squares in Parallel}
\titlerunning{Sliding Squares in Parallel}

\author{Hugo A. Akitaya}{Miner School of Computer and Information Sciences, University of Massachusetts Lowell, USA}{hugo_akitaya@uml.edu}{https://orcid.org/0000-0002-6827-2200}{}
\author{Sándor P. Fekete}{Department of Computer Science, TU Braunschweig, Germany}{s.fekete@tu-bs.de}{https://orcid.org/0000-0002-9062-4241}{}
\author{Peter Kramer}{Department of Computer Science, TU Braunschweig, Germany}{kramer@ibr.cs.tu-bs.de}{https://orcid.org/0000-0001-9635-5890}{}
\author{Saba Molaei}{Sharif University of Technology, Iran}{molaeisaba1@gmail.com}{https://orcid.org/0009-0007-8142-8904}{}
\author{Christian Rieck}{Institute of Mathematics, University of Kassel, Germany}{christian.rieck@mathematik.uni-kassel.de}{https://orcid.org/0000-0003-0846-5163}{}
\author{Frederick Stock}{Miner School of Computer and Information Sciences, University of Massachusetts Lowell, USA}{frederick_stock@student.uml.edu}{https://orcid.org/0009-0008-9005-6855}{}
\author{Tobias Wallner}{Department of Computer Science, TU Braunschweig, Germany}{wallner@ibr.cs.tu-bs.de}{https://orcid.org/0009-0004-8400-7222}{}

\authorrunning{H. A. Akitaya, S. P. Fekete, P. Kramer, S. Molaei, C. Rieck, F. Stock, and T. Wallner}
\Copyright{Hugo A. Akitaya, Sándor P. Fekete, Peter Kramer, Saba Molaei, Christian Rieck, Frederick Stock, Tobias Wallner}



\funding{%
    This work was partially supported by the German Research Foundation (DFG), project ``Space~Ants'' (FE~407/22-1) and grant 522790373, by the National Science Foundation (NSF), grant CCF-2348067, and by a fellowship of the German Academic Exchange Service~(DAAD).
}

\nolinenumbers
\begin{document}

    \maketitle

    \begin{abstract}
        We consider algorithmic problems motivated by modular robotic reconfiguration in the \newterm{sliding square model}, in which we are given $n$ square-shaped modules in a (labeled or unlabeled) start configuration and need to find a schedule of sliding moves to transform it into a desired goal configuration, maintaining connectivity of the configuration at all times.
        Recent work has aimed at minimizing the total number of moves, resulting in fully \newterm{sequential} schedules that perform reconfiguration in~$\mathcal{O}(nP)$ moves for arrangements of bounding~box perimeter size $P$~\cite{akitaya.demaine.korman.ea2022compacting-squares}, or a number of moves linear in the sum of module coordinates in the start and target arrangements~\cite{kostitsyna.ophelders.parada.ea2024optimal-in-place}.

        We extend the model to leverage the possibility of \newterm{parallel} motion, thereby reducing worst-case makespans by a factor linear in $n$.
        Our work presents tight results both in terms of complexity and algorithms:
        We show that deciding the existence of a single parallel reconfiguration step that solves an instance is \NP-complete for unlabeled modules, but can be solved efficiently in the labeled setting.
        Nevertheless, deciding whether a labeled instance can be solved in two parallel steps is \NP-complete.
        Finally, we describe an algorithm to perform in-place reconfiguration in worst-case optimal~$\mathcal{O}(P)$ parallel steps for the unlabeled setting.
        This algorithm has a straight-forward extension to the labeled setting with slight relaxations to either the reconfiguration time or space constraint.
    \end{abstract}

    \newpage
\section{Introduction}
\label{sec:introduction}

Reconfiguring an arrangement of geometric objects is a fundamental task in a wide range of areas, both in theory and practice.
A typical task arises from relocating a (potentially large) collection of agents from a given start into a desired goal configuration in an efficient manner, while avoiding collisions between objects or violating other constraints, such as maintaining connectivity of the overall arrangement.
In recent years, the problem of modular robot reconfiguration~\cite{butler.rus2003distributed-planning,rus.vona2001crystalline-robots,vassilvitskii.yim.suh2002complete-local}
has enjoyed particular attention~\cite{
    abel.akitaya.kominers.ea2024universal-in-place,
    akitaya.arkin.damian.ea2021universal-reconfiguration,
    akitaya.demaine.gonczi.ea2021characterizing-universal,
    akitaya.demaine.korman.ea2022compacting-squares,
    aloupis.collette.damian.ea2009linear-reconfiguration}
in the context of Computational Geometry: In the \newterm{sliding square model}
introduced by Fitch, Butler, and~Rus~\cite{fitch.butler.rus2003reconfiguration-planning},
a given start configuration of $n$ identical modules, each occupying a square grid cell, must be transformed by a sequence of atomic, \newterm{sequential} moves (shown in~\cref{fig:intro-legal-moves}) into a target arrangement, without losing connectivity of the underlying grid graph.

\begin{figure}[htb]
    \begin{subfigure}[t]{0.5\textwidth-0.5em}
        \hfil%
        \includegraphics[page=1]{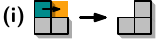}%
        \hfil%
        \includegraphics[page=2]{intro-moves}%
        \hfil%
        \subcaption{Two types of move can be performed by individual modules: \textbf{\textsf{(i)}}~Slides and \textbf{\textsf{(ii)}} convex transitions.}%
        \label{fig:intro-legal-moves}%
    \end{subfigure}%
    \hfill%
    \begin{subfigure}[t]{0.5\textwidth-1.5em}
        \hfil%
        \includegraphics[page=3]{intro-moves}%
        \subcaption{Our model allows moves to be performed in parallel transformations, as shown here.}
        \label{fig:intro-chain-moves}
    \end{subfigure}%
    \caption{Our model allows for two types of moves to occur in parallel, collision-free transformations. In this paper, we show the symmetric difference of a transformation using turquoise and yellow.}
\end{figure}%

Aiming at minimizing the total number of moves, previous research has resulted in considerable progress, recently establishing~\cite{akitaya.arkin.damian.ea2021universal-reconfiguration} universal configuration in $\mathcal{O}(nP)$ moves for a 2-dimensional arrangement of $n$ modules with bounding box perimeter size $P$, and $\mathcal{O}(n^2)$ in three dimensions~\cite{abel.akitaya.kominers.ea2024universal-in-place, kostitsyna.ophelders.parada.ea2024optimal-in-place}.
However, the resulting schedules are purely sequential, not optimally minimizing the overall time until completion, called \newterm{makespan}.
With parallel motion, much lower makespan can be achieved---which is also a more challenging objective, as it requires
coordinating the overall motion plan, not just at the atomic level (see~\cref{fig:intro-chain-moves}), but also at the global level to maintain connectivity and avoid collisions.

\subsection{Our contributions}
\label{subsec:our-contributions}
We provide first results for \emph{parallel} reconfiguration in the sliding squares model with no restriction on the input.
For~precise definitions and the problem statement, see~\cref{sec:preliminaries}.
Our central results can be summarized as follows, starting with the \emph{unlabeled} setting:

\begin{enumerate}
    \item Deciding whether a single transformation can solve an instance is \NP-complete in general, and \NP-hard to approximate by a factor better than $2$.\footnote{%
        In an extended abstract that appeared at ESA 2025, we misstated this as implying \APX-hardness.
        Instead, this shows that {\textsc{Parallel Sliding Squares}} and its {\textsc{Labeled}} variant are not in \PTAS.
    }
    \item In-place schedules of (asymptotically worst-case optimal) makespan $\mathcal{O}(P)$ can be computed efficiently, where $P$ is the perimeter of the union the configurations' bounding boxes.
\end{enumerate}

We extend these results to the \emph{labeled}~variant, in which individual modules are distinguishable and have unique start and target positions.
In particular, we show the following.

\begin{enumerate}
    \setcounter{enumi}{2}
    \item Deciding whether an instance can be solved in makespan $1$ can be done in $\mathcal{O}(n)$, but is~\NP-complete for makespan $2$, and \NP-hard to approximate by a factor better than~$\nicefrac{3}{2}$.%
    \footnotemark[\value{footnote}]
    \item Our algorithmic methods from the unlabeled setting can be extended to compute worst-case optimal schedules, at the cost of relaxing the in-place constraint.
\end{enumerate}

\subsection{Related work}
\label{subsec:related-work}

On the theoretical side, algorithmic methods for coordination the motion of many robots can be traced back to the classical work of Schwartz and Sharir~\cite{schwartz.sharir1983on-piano} from the 1980s.
On the practical side,~\cite{fukuda.nakagawa.kawauchi.ea1989structure-decision} presented an architecture for modular robots, followed by a wide spectrum of work that often used cuboids as elementary building blocks; see~\cite{yim.shen.salemi.ea2007modular-self-reconfigurable} for a survey.

Of particular interest for the algorithmic side is the \newterm{sliding cube model} (or \emph{square}) by Fitch, Butler, and Rus~\cite{fitch.butler.rus2003reconfiguration-planning}, introduced in the context of a modular robotic hardware~\cite{butler.rus2003distributed-planning,rus.vona2001crystalline-robots,vassilvitskii.yim.suh2002complete-local}.
Recent work~\cite{abel.akitaya.kominers.ea2024universal-in-place,akitaya.demaine.korman.ea2022compacting-squares,dumitrescu.pach2006pushing-squares,kostitsyna.ophelders.parada.ea2024optimal-in-place}
has studied algorithmic methods for sequentially sliding squares and cubes, with typical schedules requiring a quadratic number of moves.
Also related are models with slightly different types of moves, such as ``pivoting'', see~\cite{akitaya.arkin.damian.ea2021universal-reconfiguration,akitaya.demaine.gonczi.ea2021characterizing-universal,connor.michail.skretas2024all-for,michail.skretas.spirakis2019on-transformation}.

In~\cite{michail.skretas.spirakis2019on-transformation} the authors claim an algorithm for connectivity-preserving reconfiguration of sliding squares in $\mathcal{O}(n)$ parallel steps.
Unfortunately, this is incorrect for general input.
Nevertheless, if the input configurations are guaranteed to not contain a specific ``forbidden'' local pattern (a $2\times 2$ arrangement with only two diagonally adjacent squares), then their algorithm is correct.
Similar models of parallel sliding squares have been investigated experimentally~\cite{wolters2024parallel-algorithms}.
However, this still leaves the existence of linear-time parallel protocols unresolved.

While the previous work on sliding squares focuses on minimizing the number of moves, a different line of research aims at minimizing the total time until completion.
In~\cite{hurtado.molina.ramaswami.ea2015distributed-reconfiguration}, the authors study parallel distributed algorithms for a less restrictive model of sliding squares, obtaining $\mathcal{O}(n)$ makespan.
They show how to obtain the same result in the standard sliding square model using \newterm{meta-modules}, small groups of cooperating modules that behave like a single entity.
This, however, imposes extra constraints in the input.
Algorithmic research on meta-modules was considered in~\cite{aloupis.collette.damian.ea2009linear-reconfiguration,parada.sacristan.silveira2021new-meta-module}.
In~\cite{becker.fekete.keldenich.ea2018coordinated-motion,demaine.fekete.keldenich.ea2019coordinated-motion}, the authors considered coordinated motion planning in a less constrained model (see discussion in~\cref{sec:preliminaries}) for labeled robots to
minimize the makespan, aiming for \newterm{constant stretch}, i.e., a makespan
within a constant factor of the maximum distance
for a single robot.
This was extended to preserve connectivity for unlabeled~\cite{bourgeois.fekete.kosfeld.ea2022space-ants,fekete.keldenich.kosfeld.ea2023connected-coordinated} and labeled robots~\cite{fekete.kramer.rieck.ea2024efficiently-reconfiguring}, and between obstacles~\cite{fekete.kosfeld.kramer.ea2024coordinated-motion}.

\subparagraph*{Scaffolding and meta-modules.}
Computing reconfiguration paths is not a simple task and, for many models, it is computationally intractable (even \PSPACE-complete~\cite{akitaya.demaine.gonczi.ea2021characterizing-universal}).
\newterm{Scaffolding} (proposed independently by~\cite{kotay.rus2000algorithms-for} and~\cite{nguyen.guibas.yim2001controlled-module}) aims at making reconfiguration simpler by restricting the scope to configurations containing a regular substructure~(the \newterm{scaffold}) that remains connected between moves (thus removing complexity from the connectivity constraint) while having enough empty positions to allow modules to move through the scaffold (thus removing complexity from the free space constraints).
They obtain a scaffold by using \newterm{meta-modules}, groups of modules that cooperate to move as a higher-level functional unit.
Previous work has used these to translate algorithms for a particular model to another (see, e.g.,~\cite{hurtado.molina.ramaswami.ea2015distributed-reconfiguration,parada.sacristan.silveira2021new-meta-module,rus.vona2001crystalline-robots}).

The drawback of these techniques is that their input must belong to the subset of configurations where modules are already organized into meta-modules (the scaffold exists in both the initial and target configuration).
The authors of~\cite{bourgeois.fekete.kosfeld.ea2022space-ants,fekete.keldenich.kosfeld.ea2023connected-coordinated} proposed an algorithm that first builds the scaffold from an arbitrary scaled configuration.
They achieve \newterm{constant stretch}, i.e., a makespan within a constant factor of the maximum minimum distance for a single robot (based on previous work~\cite{becker.fekete.keldenich.ea2018coordinated-motion,demaine.fekete.keldenich.ea2019coordinated-motion}).
This was extended for labeled modules~\cite{fekete.kramer.rieck.ea2024efficiently-reconfiguring}, and between obstacles~\cite{fekete.kosfeld.kramer.ea2024coordinated-motion}.
Although their techniques apply to a wider range of configurations, there are still restrictions on the input.
The algorithm presented in this paper also makes use of scaffolding and meta-module techniques.
To the best of our knowledge, our algorithm is the first universal one to build the scaffold and meta-modules (i.e., from arbitrary input).

\subparagraph*{Model comparison.}
Our model is more constrained than existing models for parallel reconfiguration under connectivity constraints~\cite{dumitrescu.suzuki.yamashita2004motion-planning,fekete.keldenich.kosfeld.ea2023connected-coordinated,fekete.kramer.rieck.ea2024efficiently-reconfiguring,hurtado.molina.ramaswami.ea2015distributed-reconfiguration,michail.skretas.spirakis2019on-transformation}.
Michail, Skretas, and Spirakis~\cite{michail.skretas.spirakis2019on-transformation} do not explicitly require the connected backbone constraint, and do not explicitly define the collision model.
The models in~\cite{fekete.keldenich.kosfeld.ea2023connected-coordinated,fekete.kramer.rieck.ea2024efficiently-reconfiguring} do not require the connected backbone constraint, and allow modules to enter and leave cells simultaneously, even in orthogonal directions.
Similarly,~\cite{hurtado.molina.ramaswami.ea2015distributed-reconfiguration} relaxes the free-space constraint of the convex transition move, allowing a module to move through two diagonally adjacent static modules.
However, the authors require that the cells involved in individual slide moves in the same transformation are disjoint, disallowing moves as shown in~\cref{fig:intro-chain-moves}.
The same is true for~\cite{dumitrescu.suzuki.yamashita2004formations-for,dumitrescu.suzuki.yamashita2004motion-planning} where the paths of transformations at any given time stamp must be disjoint.
%
\section{Preliminaries}
\label{sec:preliminaries}
We study the connected reconfiguration of squares in the two-dimensional integer grid by sliding moves in a parallel variant of the sequential \newterm{sliding squares} model, as follows.

\subparagraph*{Configurations.}
Each simple $4$-cycle of edges in the infinite integer grid $\mathbb{Z}^2$ bounds a unit \newterm{cell} $c$, which is uniquely identified by its minimal integer coordinates~$x(c)$ and $y(c)$.
Two such cells are \newterm{edge-adjacent} (resp., \newterm{vertex-adjacent}) if their boundary cycles share an edge (resp., vertex).
A~\newterm{configuration} $C$ of $n$ square \newterm{modules} is defined by the set of occupied cells, each containing a single module.
We say that $C$ is valid if and only if its dual graph, defined by edge-adjacency, is connected, and denote by $B$ the unique axis-aligned bounding box of minimal perimeter~$P$ that contains $C$; see~\cref{fig:intro} for illustrations.

\begin{figure}[htb]
    \hfil%
    \begin{subfigure}[t]{0.26\columnwidth}
        \hfil%
        \includegraphics[page=1]{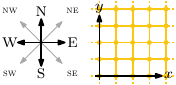}%
        \subcaption{The integer grid $\mathbb{Z}^2$.}
        \label{fig:intro-grid}%
    \end{subfigure}%
    \hfil%
    \begin{subfigure}[t]{0.72\columnwidth}
        \hfil%
        \includegraphics[page=2]{intro-grid-configuration}%
        \hfil%
        \includegraphics[page=3]{intro-grid-configuration}%
        \hfil%
        \includegraphics[page=4]{intro-grid-configuration}%
        \subcaption{A configuration $C$ of $10$ modules, its dual graph, and bounding box $B$.}
        \label{fig:intro-configuration}
    \end{subfigure}%
    \caption{We illustrate the relation of the integer grid and configurations.}%
    \label{fig:intro}
\end{figure}%

\subparagraph*{Moves.}
Individual modules can perform two types of \newterm{move} into an unoccupied edge-adjacent or vertex-adjacent cell, \newterm{slides} and \newterm{convex transitions}, as illustrated in~\cref{fig:intro-legal-moves}.
In our model, both moves take an identical (unit) duration to complete, allowing us to easily extend the existing notion of move counting in sequential models~\cite{akitaya.demaine.korman.ea2022compacting-squares,dumitrescu.pach2006pushing-squares,dumitrescu.suzuki.yamashita2004motion-planning,moreno.sacristan2020reconfiguring-sliding} to the parallel setting.
Note that in our figures, convex transitions are shown as paths containing a circular arc for clarity.
Such a path would be accurate if the corners of modules were rounded.

\subparagraph*{Transformations.}
The parallel execution of moves constitutes a \newterm{transformation}, which we denote by the initial and resulting configurations, e.g., $C_1\rightarrow C_2$.
A transformation is valid only if it (i) preserves connectivity and (ii) does not cause collisions.

As in a number of existing models for both sequential and parallel reconfiguration of squares in the grid~\cite{abel.akitaya.kominers.ea2024universal-in-place,akitaya.demaine.korman.ea2022compacting-squares,kostitsyna.ophelders.parada.ea2024optimal-in-place,wolters2024parallel-algorithms}, connectivity is determined by the existence of a \newterm{connected backbone}:
Given a configuration $C$, we call a subset $M$ of modules \newterm{free} if $C\setminus M'$ is a valid configuration for any $M'\subseteq M$.
Let $M\subset C_1$ refer to all moving modules during $C_1\rightarrow C_2$.
Then there exists a connectivity-preserving backbone if and only if $M$ is free both in $C_1$ and in $C_2$.
A~module $m$ is \newterm{free} if $\{m\}$ is free.
In~\cref{fig:intro-backbone}, both transformations are illegal as they attempt to move modules that are not free (conflicts are marked in red).
\begin{figure}[htb]
    \centering%
    \transformable{\includegraphics[page=1]{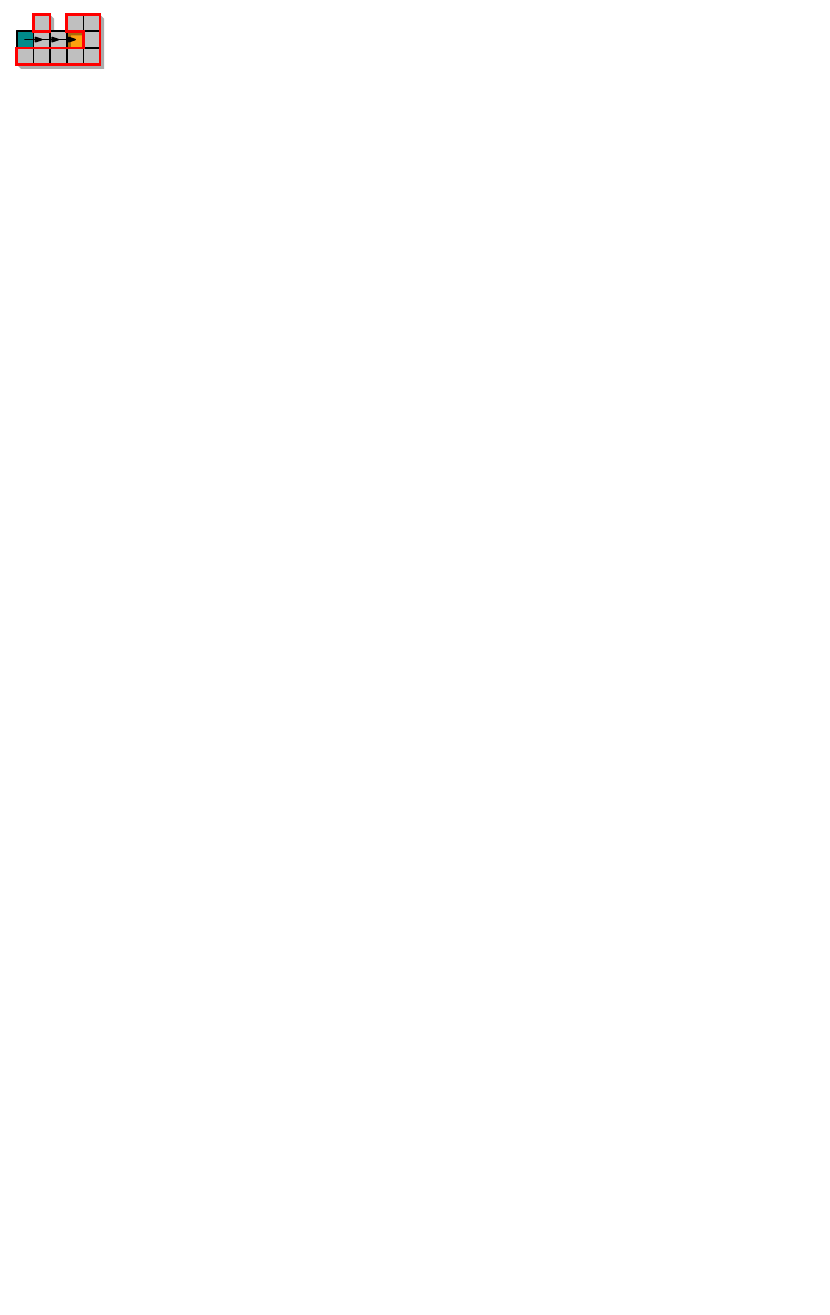}}%
    \ntransforms%
    {\includegraphics[page=2]{intro-backbone}}%
    \ntransforms%
    {\includegraphics[page=3]{intro-backbone}}%
    \caption{Illegal transformations: Moving non-free modules causes disconnections.}%
    \label{fig:intro-backbone}%
\end{figure}%

Secondly, a transformation is collision-free exactly if all modules can move along their designated path at a constant rate, completing it in unit time without overlapping with any other module in the process.
We illustrate a partial selection of collisions in~\cref{fig:intro-collisions}.
Similar to the sequential variant, the defined moves and transformations are reversible, i.e., if for any transformation $C_1\rightarrow C_2$, there exists an inverse transformation~${C_2\rightarrow C_1}$.

\begin{figure}[htb]
    \hfil%
    \includegraphics[page=4]{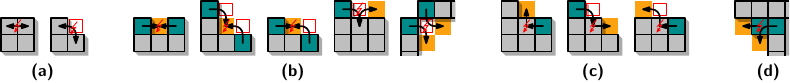}%
    \caption{Collisions: (a) Modules cannot swap places or (b) enter the same cell simultaneously, (c)~moves cannot meet at endpoints orthogonally, and (d) speed mismatches must be avoided.}
    \label{fig:intro-collisions}%
\end{figure}

\subparagraph*{Problem statement.}
The {\textsc{Parallel Sliding Squares}} problem asks for the connected reconfiguration of modules in the infinite integer grid by sequences of legal transformations.
Such a sequence is referred to as a \newterm{schedule}, and its length as its \newterm{makespan}.
An~instance of {\textsc{Parallel Sliding Squares}} asks, for two configurations $C_1,C_2$ of $n$ modules and some~$k\in\mathbb{N}^+$, whether there exists a schedule $C_1 \rightarrow^* C_2$ of makespan at most~$k$.

\subparagraph*{Labeled problem variant.}
In the case of distinguishable (or \newterm{labeled}) modules, we speak of the {\textsc{Labeled Parallel Sliding Squares}} problem.
This variant uses the exact same move set and collision constraints as the unlabeled variant, but configurations are now injective mappings from $[n]$ to $\mathbb{Z}^2$ rather than subsets of the latter.
This allows us to specify individual target positions and track the movement of individual modules across transformations.

\subparagraph*{Useful notation.}
Throughout this paper, we make use of cardinal and ordinal directions.
In particular, the unit vector $(1,0)$ points \newterm{east}, $(0,1)$ points \newterm{north}, and their opposite vectors point \newterm{west} and \newterm{south}, respectively.
For an illustration, see~\cref{fig:intro-grid}.
Throughout this paper, we use $B_1$ and $B_2$ to denote the bounding boxes of $C_1$ and~$C_2$, respectively.
We assume, as in the literature~\cite{moreno.sacristan2020reconfiguring-sliding}, that $B_1$ and $B_2$ share a south-west corner at $(0,0)$.

We define the (open) neighborhood $N(c)$ of a cell $c$ as the set of cells that are edge-adjacent to $c$, and the closed neighborhood $N[c]$ as $N(c)\cup\{c\}$.
We abuse notation to express neighborhoods of sets of cells as $N[S]=\bigcup_{c\in S}N[c]$ and $N(S)=N[S]\setminus S$.
Analogously, define $N^*(S)$ (resp., $N^*[S]$) as the open (resp., closed) neighborhood using vertex-adjacency.

Given a collection of modules~$M$, let $M[i]^x=\{m\in M \mid x(m)=i\}$ be the subset of~$M$ of modules with $x$ coordinate $i$.
Likewise let $M[i,j]^x$ be the set of all modules within the $x$-interval~$[i,j]$, where $j\geq i$.
The same selectors along the $y$-axis are denoted by $M[\cdot]^y$.

A schedule for $\mathcal{I}$ is (\newterm{strictly}) \newterm{in-place} (as defined in~\cite{akitaya.demaine.korman.ea2022compacting-squares,moreno.sacristan2020reconfiguring-sliding}) if and only if no intermediate configuration exceeds the union $B_1\cup B_2$ by more than one module.
We relax this constraint slightly and say that a schedule is \newterm{weakly in-place} if no intermediate configuration exceeds the union $B_1\cup B_2$ by more than a constant number of units.
Note that we do not restrict the number of modules located outside the union, only their distance to it.
This~corresponds to a slightly more restricted variant of the definition used in~\cite{abel.akitaya.kominers.ea2024universal-in-place}.
%

\section{Computational complexity}
\label{sec:hardness}

We provide a number of complexity results for {\textsc{Parallel Sliding Squares}}, both in the labeled and unlabeled variant, which are complementary to the \NP-completeness result obtained by the authors of~\cite{akitaya.demaine.korman.ea2022compacting-squares} for the sequential model.

We start by briefly arguing that (\textsc{Labeled}) {\textsc{Parallel Sliding Squares}} is in \NP.

\begin{proposition}
    {\textsc{Parallel Sliding Squares}} and its {\textsc{Labeled}} variant are in \NP.
    \label{prop:sliding-in-np}
\end{proposition}
\begin{proof}
    We assume real RAM; a labeled (or unlabeled) configuration of $n$ modules is then simply ordered sequence coordinates in $\mathbb{Z}^2$, i.e., $\mathcal{O}(n)$ space.
    An instance of either variant consists of two configurations and a single integer $k\in\mathbb{Z}$ for the maximum makespan.

    To verify a solution, consider as a witness the schedule $C_1 \rightarrow^k C_{k+1}$.
    Every transformation in this schedule moves at most $n$ modules, so this witness requires $\mathcal{O}(nk)$ space.
    Note that, expressed as such, the witness size is unbounded in the input size $n$.
    However, due to~\cite{dumitrescu.pach2006pushing-squares}, there exists a schedule of $\mathcal{O}(n^2)$ sequential moves for any given instance, i.e., there exists a constant $c\in\mathbb{R}^+$ such that every instance with $k\geq cn^2$ is a yes-instance.
    We may thus assume, without loss of generality, that $k\in\mathcal{O}(n^2)$ and so, $\mathcal{O}(nk)\subset\mathcal{O}(n^3)$.

    To verify the witness, we proceed iteratively, verifying the transformations in the schedule one by one.
    Let $M_i\subset C_i$ refer to the set of modules that move during $C_i\rightarrow C_{i+1}$.
    We can verify that $M_i$ is free in $C_i$ by checking whether $C_i\setminus M_i$ forms a connected configuration, which takes $\mathcal{O}(n)$ time using breadth-first search.
    By performing at most a constant number of checks in the proximity of each moving module $m\in M_i$, we can verify that its move (a)~does not collide with any other move and (b) is supported by backbone modules in $C\setminus M_i$.
    Using dynamic perfect hashing, we can represent the current configuration in $\mathcal{O}(n)$ space such that each check takes $\mathcal{O}(1)$ time~\cite{dietzfelbinger.karlin.ea1994dynamic-perfect}, i.e., $\mathcal{O}(n)$ per transformation.

    Once all transformations have been checked, we verify that the resulting configuration matches the target configuration of the instance.
    The total runtime is then $\mathcal{O}(nk)$.
\end{proof}

In the following two subsections, we give detailed proof of complexity for both the labeled and unlabeled variants of {\textsc{Parallel Sliding Squares}}.
We reduce from the \NP-complete problem {\textsc{Planar Monotone 3Sat}}~\cite{berg.khosravi2012optimal-binary}, which asks whether a Boolean formula in conjunctive normal form is satisfiable.
Each clause consists of at~most~$3$~literals, all either positive or negative, and the clause-variable incidence graph must admit a plane drawing where variables are mapped to the $x$-axis, positive (resp., negative) clauses are mapped to the upper (resp., lower) half-plane, and edges do not cross the $x$-axis.

\subsection{Complexity of the unlabeled variant}
\label{subsec:complexity-unlabeled}

\begin{theorem}
    \label{thm:unlabeled-makespan-hard}
    {\textsc{Parallel Sliding Squares}} is \NP-complete for makespan $1$.
\end{theorem}

Our reduction starts with a rectilinear embedding of the clause-variable incidence graph of an instance $\varphi$ of {\textsc{Planar Monotone 3Sat}} and constructs an instance~$\mathcal{I}_\varphi$ of \textsc{Parallel Sliding Squares} using \newterm{variable} and \newterm{clause gadgets}.
We then argue that~$\varphi$ can be satisfied if and only if there is a single transformation that solves $\mathcal{I}_\varphi$.

For a formula $\varphi$ over the variables $x_1,x_2,\ldots, x_m$, we create $m$ copies of the variable gadget in a horizontal line, directly adjacent to one another.
For each positive (negative) clause, we further place a corresponding gadget connected to its respective variables from above (below) the $x$-axis.
A~simple example of our construction is depicted in~\cref{fig:hardness-unlabeled-highlevel}.
\begin{figure}[htb]
    \begin{subfigure}[t]{\columnwidth}
        \hfil%
        \includegraphics[page=1,width=\columnwidth]{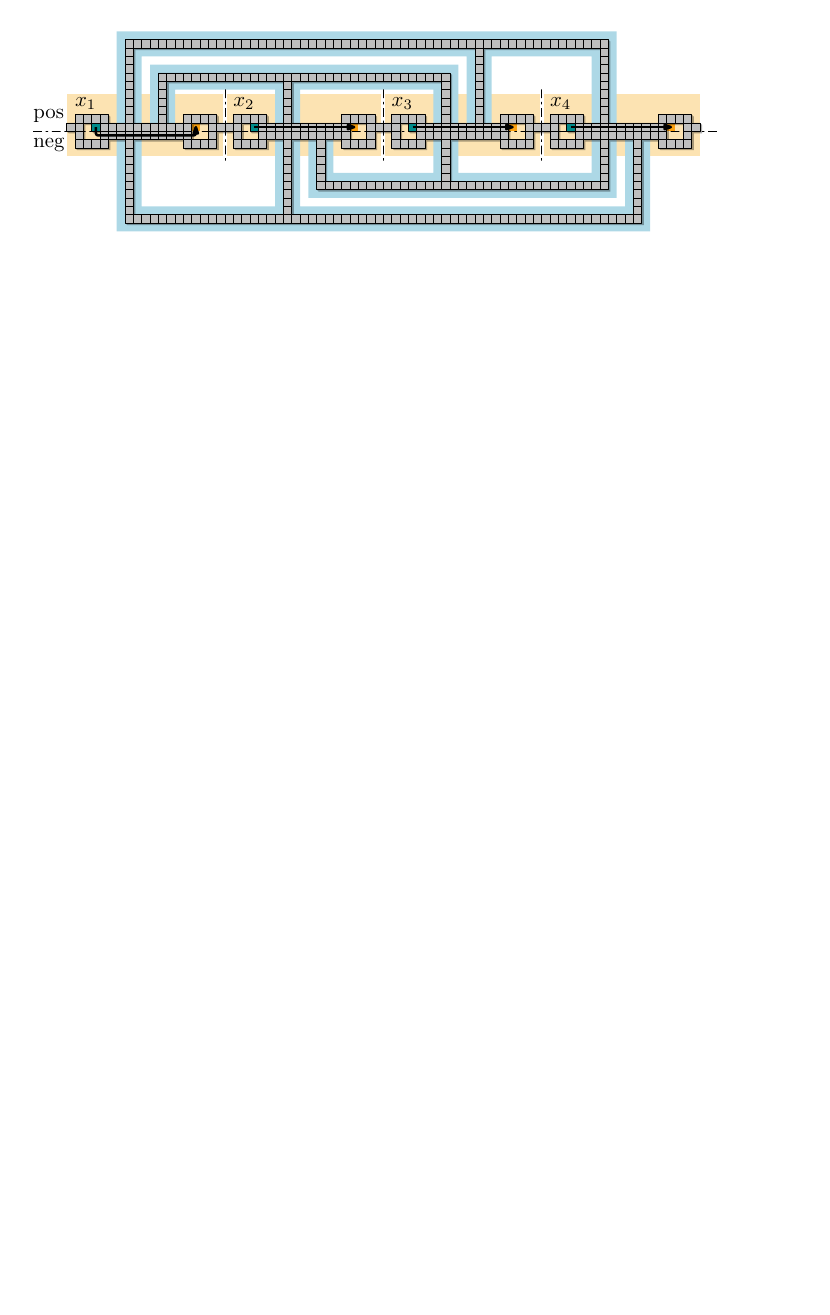}%
        \subcaption{Our construction for $\varphi = (x_1 \lor x_3 \lor x_4)\land(x_1 \lor x_2 \lor x_3)\land(\overline{x_1}\lor\overline{x_2}\lor\overline{x_4})\land(\overline{x_2}\lor\overline{x_3}\lor\overline{x_4})$.
        The depicted transformation represents the satisfying assignment $\alpha(\varphi) = (\texttt{true},\ \texttt{false},\ \texttt{false},\ \texttt{false})$.}
        \label{fig:hardness-unlabeled-highlevel}
    \end{subfigure}\par\medskip%
    \captionsetup[subfigure]{justification=centering}%
    \begin{subfigure}[t]{0.5\textwidth-1em}%
        \hfil%
        \includegraphics[page=2]{hardness-unlabeled}%
        \subcaption{Variable gadget (blue/red $\equiv$ \texttt{true}/\texttt{false}).}
        \label{fig:unlabeled-variable}
    \end{subfigure}%
    \hfill%
    \begin{subfigure}[t]{0.5\textwidth-1em}%
        \hfil%
        \includegraphics[page=3]{hardness-unlabeled}%
        \subcaption{Clause gadget.}
        \label{fig:unlabeled-clause}
    \end{subfigure}%
    \caption{An overview of our hardness reduction and the two types of gadget used.}
\end{figure}

\subparagraph*{Variable gadget.}%
The variable gadget consists of two cycles of $12$ modules each, the left containing an additional module in its interior, see~\cref{fig:unlabeled-variable}.
These cycles are connected by a horizontal \newterm{assignment strip} of height $2$, to which individual clause gadgets (blue) are connected.
In the target configuration, the left cycle is empty, and an extra module is now located in the right.
We highlight two transformation representing either a \texttt{true} or \texttt{false} value assignment, respectively, in~\cref{fig:unlabeled-variable}.
Each option locally severs the edge-adjacency between the variable gadget and all its positive or negative clauses, respectively.

\subparagraph*{Clause gadget.}%
The clause gadget consists of a thin horizontal strip that spans the gadgets of variables contained in the clause, and (up to) three vertical prongs that connect to the assignment strips of the incident variable gadgets, see~\cref{fig:unlabeled-clause}.
Note that there is not difference between the start and target configuration.

\begin{proof}[Proof of~\cref{thm:unlabeled-makespan-hard}]
    We reduce from {\textsc{Planar Monotone 3Sat}}.
    For a given instance $\varphi$ of {\textsc{Planar Monotone 3Sat}}, we start by constructing a {\textsc{Parallel Sliding Squares}} instance $\mathcal{I}_\varphi$ by means of the previously outlined gadgets.
    We now show that $\varphi$ has a satisfying assignment $\alpha(\varphi)$ if and only if there exists a schedule of makespan $1$ that solves $\mathcal{I}_\varphi$.

    \begin{claim}
        \label{clm:hardness-assignment-implies-schedule}
        If $\varphi$ has a satisfying assignment, there is a single transformation that solves $\mathcal{I}_\varphi$.
    \end{claim}
    \begin{claimproof}
        Given a satisfying assignment $\alpha(\varphi)$ for $\varphi$, let $\alpha(x_i)\in\{\texttt{true}, \texttt{false}\}$ refer to the Boolean value assigned to the variable $x_i$.
        We solve each variable gadget based on $\alpha(x_i)$.

        In case that $\alpha(x_i)=\texttt{true}$, we move the modules in the ``negative'' side of the assignment strip as indicated by the blue path in~\cref{fig:unlabeled-variable}.
        This breaks locally disconnects all incident negative clauses, but preserves a backbone consisting of the remaining modules within the gadget and all  the incident positive clause gadgets.
        A transformation involving only sliding moves of modules along the ``positive'' side handles the negative case.
        Parallel application in each variable gadget provides us with a joint, connected backbone along all gadgets.

        It remains to argue that all clause gadgets retain a connection to this backbone, i.e., that there exists a single connected backbone.
        This is the case if and only if each clause gadget is incident to a variable gadget with a matching Boolean value assignment.
        Choosing our parallel transformations based on $\alpha(\varphi)$ therefore yields a schedule of makespan $1$ for $\mathcal{I}_\varphi$.
    \end{claimproof}

    \begin{claim}
        \label{clm:hardness-schedule-implies-assignment}
        If there is a transformation that solves $\mathcal{I}_\varphi$, then $\varphi$ has a satisfying assignment.
    \end{claim}

    \begin{claimproof}
        All cells that are occupied in both configurations of $\mathcal{I}_\varphi$ must either retain the same module, or, if the module moves, be entered by another module simultaneously.
        A transformation that solves the instance thus induces a set of directed paths (``chains'') that connect pairs of positions that are occupied in only one of the two configurations.

        Recall that the symmetric difference contains exactly $2m$ cells, all located within variable gadgets, where $m$ is the number of variables in $\varphi$.
        It directly follows that the chain moves are limited to a single variable gadget each:
        The modules marked with a cross in~\cref{fig:unlabeled-variable} cannot move to either of their occupied neighbor cells.

        It remains to show that a transformation solving $\mathcal{I}_\varphi$ always implies a satisfying assignment for $\varphi$.
        To this end, we show that such a transformation always disconnects either all incident positive or negative clauses from a variable gadget.
        We start with the excess square in the left circle of the variable gadget, which can reach exactly two cells that are occupied in the target configuration, either the $\east$- or $\south\east$-adjacent cell.
        From this point onward, due to various collision restrictions, modules can only move towards the right circle in a straight line, disconnecting the assignment strip from adjacent prongs during its transformation.

        We conclude that a transformation solving a variable gadget causes either all positive or negative clause prongs to lose connection to the variable throughout.
        Assigning each variable a Boolean value based on the direction taken by the excess square in its variable gadget therefore yields a satisfying assignment for $\varphi$.
    \end{claimproof}
    
    \cref{clm:hardness-assignment-implies-schedule,clm:hardness-schedule-implies-assignment} conclude the proof.
\end{proof}

\cref{thm:unlabeled-makespan-hard} highlights a previously unrecognized complexity gap compared to related models for parallel connected reconfiguration without the connected backbone constraint.
In particular, in the unlabeled model studied in~\cite{fekete.keldenich.kosfeld.ea2023connected-coordinated,fekete.kramer.rieck.ea2024efficiently-reconfiguring}, deciding the existence of schedules of makespan $2$ is \NP-complete, while makespan $1$ is in~\P.

\begin{corollary}
    \label{cor:unlabeled-apx-hardness}
    Unless $\P=\NP$, {\textsc{Parallel Sliding Squares}} and cannot be approximated within a factor better than $2$ in polynomial time.
\end{corollary}

\begin{proof}
    To prove this, it suffices to show that any instance $\mathcal{I}_\varphi$ created by our hardness reduction can be solved in $2$ parallel transformations without deciding the underlying satisfiability problem.
    Recall that the two configurations differ only in the variable gadgets.

    We describe two transformations that reconfigure every variable gadget from its initial to its target configuration.
    Assume, without loss of generality, that the variable gadgets are sorted in increasing order by index; otherwise, we may relabel them accordingly.
    The reconfiguration proceeds as follows: (1) we first perform a sliding chain move along the ``positive'' part of each variable gadget's assignment strip, ensuring that each clause gadget remains connected to the variable of highest index among those it contains; (2) we then perform a second sliding chain move involving the remaining modules of the upper part of every variable, yielding the target configuration.
    As during the second transformation, each clause remains connected to at least one of its variables (specifically, the one with the smallest index that contains it), the instance admits a schedule of makespan $2$.
\end{proof}

\subsection{Complexity of the labeled variant}
\label{subsec:complexity-labeled}

For the {\textsc{Labeled Parallel Sliding Squares}} problem, in which individual modules are distinguishable, we show that it can be decided in polynomial time whether a schedule of makespan $1$ exists, while the same is \NP-complete for makespan $2$.
\begin{corollary}
    \label{cor:labeled-makespan-one}
    {\textsc{Labeled Parallel Sliding Squares}} with makespan $1$ can be decided in~$\mathcal{O}(n)$ time and space.
\end{corollary}
\begin{proof}
    This directly follows from \cref{prop:sliding-in-np}.
    As both configurations are labeled, we can simply rewrite their difference as a witness of size $\mathcal{O}(n)$ and pass it to our verifier, which then decides the instance in $\mathcal{O}(n)$ time.
\end{proof}

For makespan $2$, our reduction closely follows the structure in the previous section.
With minor changes to both gadgets and a new, third type of gadget, we prove the following.
\begin{theorem}
    \label{thm:labeled-hardness}
    {\textsc{Labeled Parallel Sliding Squares}} is \NP-complete for makespan~$2$.
\end{theorem}

\subparagraph*{Variable gadget.} The difference between the labeled and unlabeled variable gadgets is mostly in the fact that \emph{two} modules must now be moved from one side to the other.
In fact, the entire assignment strip must now be translated by one unit to the east.
Otherwise, the structure remains largely unchanged, albeit labeled, see~\cref{fig:hardness-labeled-highlevel}.
\begin{figure}[t]
    \begin{subfigure}[t]{\columnwidth}%
        \includegraphics[page=1,width=\columnwidth]{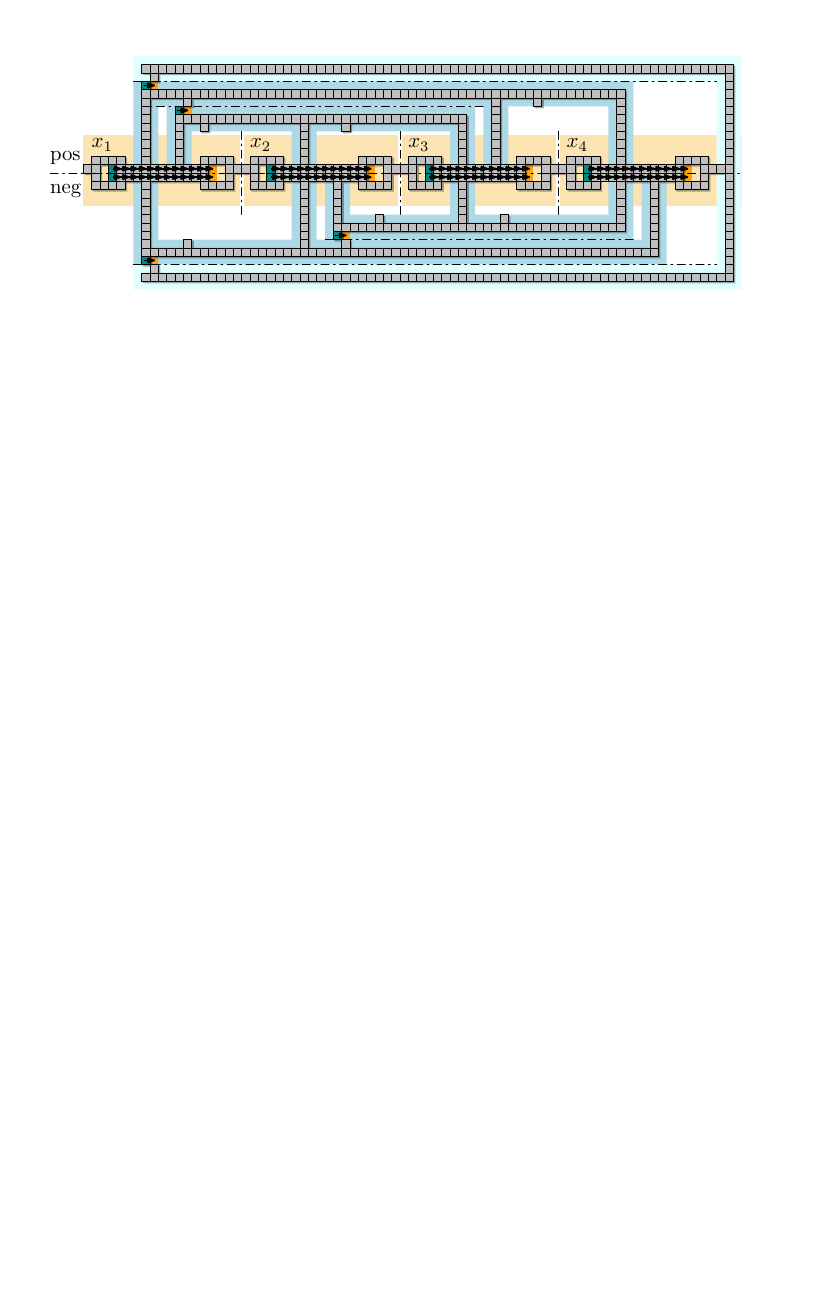}
        \caption{Our instance construction for the labeled variant of the problem, using the same formula as~\cref{fig:hardness-unlabeled-highlevel}.
        Note the newly introduced backbone gadget (light blue) above and below the instance's clause gadgets.}
        \label{fig:hardness-labeled-highlevel}
    \end{subfigure}%
    \par\medskip%
    {\captionsetup[subfigure]{justification=centering}%
        \begin{subfigure}[t]{0.5\columnwidth}%
            \hfil%
            \includegraphics[page=2]{hardness-labeled}
            \subcaption{Variable gadget and labeling.}
            \label{fig:labeled-variable}
        \end{subfigure}%
    }%
    \begin{subfigure}[t]{0.5\columnwidth}%
        \hfil%
        \includegraphics[page=3]{hardness-labeled}
        \subcaption{Clause gadget; note the extra module at the top.}
        \label{fig:labeled-clause}
    \end{subfigure}%
    \caption{An overview of our  reduction from \textsc{Planar Monotone 3Sat} to
            {\textsc{Labeled Parallel Sliding Squares}}; arrows indicate the assigned target position of a module.}
\end{figure}

\subparagraph*{Clause and backbone gadgets.}
Apart from one module that needs to move one step to the right, each clause gadget is identical in both configurations of the instance and resembles the unlabeled gadget, compare~\cref{fig:unlabeled-clause,fig:labeled-clause}.
These additional squares guarantee connectivity of the configuration during the second transformation step, which cannot be provided by the variable gadgets alone.
To make this gadget work, we introduce an additional structure, the \newterm{backbone gadget}.
It spans the entire configuration above the topmost and below the bottommost clause that connects everything to the variables, as visualized in~\cref{fig:hardness-labeled-highlevel}.

\begin{proof}[Proof of~\cref{thm:labeled-hardness}]
    We again reduce from {\textsc{Planar Monotone 3Sat}}, i.e., we construct an instance~$\mathcal{I}_\varphi$ of {\textsc{Labeled Parallel Sliding Squares}} from any given formula~$\varphi$ by using the described gadgets.
    We show that $\varphi$ has a satisfying assignment $\alpha(\varphi)$ if and only if there exists a schedule of makespan $2$ that reconfigures $\mathcal{I}_\varphi$.
    \begin{claim}
        \label{clm:labeled-hardness-assignment-implies-schedule}
        If $\varphi$ has a satisfying assignment $\alpha(\varphi)$, there is a schedule of makespan $2$ for $\mathcal{I}_\varphi$.
    \end{claim}
    \begin{claimproof}
        Consider the satisfying assignment $\alpha(\varphi)$ for $\varphi$, and let $\alpha(x_i)$ refer to the Boolean value assigned to $x_i$, i.e., $\texttt{true}$ or $\texttt{false}$.
        We argue based on $\alpha(\varphi)$, as follows.
        Just as in the unlabeled variant, according to the assignment, for every $x_i$ we can perform the respective chain move, as illustrated in~\cref{fig:labeled-variable}.
        In case that $\alpha(x_i) = \texttt{true}$, this chain move breaks connectedness with incident negative clauses.
        Note that both chains of the assignment strip for a variable cannot move simultaneously.
        Furthermore, we also move all squares associated with the backbone in this very first move, as depicted in~\cref{fig:labeled-clause}.
        As $\alpha(\varphi)$ satisfies~$\varphi$, this first parallel transformation maintains a connected configuration.

        Afterward, we perform all other chain moves within the variable gadgets.
        As the connected backbone remains stationary during the second transformation, the whole configuration stays connected.
        Thus, the desired target configuration is reached after two parallel moves.
    \end{claimproof}

    \begin{claim}
        \label{clm:labeled-hardness-schedule-implies-assignment}
        If there is a schedule of makespan $2$ for $\mathcal{I}_\varphi$, then $\varphi$ has a satisfying assignment.
    \end{claim}
    \begin{claimproof}
        In any schedule of makespan $2$, both horizontal chains of an assignment strip can not move simultaneously, as this would require illegal transformations.
        However, in each of the two transformation steps, exactly one of the chains has to move.
        As we require a connected backbone during any transformation, we use this movement to assign the respective values to each associated variable.
        This assignment satisfies the given formula: Every clause gadget remains connected to at least one variable gadget of matching value.
        Similar to the unlabeled variant, there are no other feasible moves within variable gadgets that guarantee both connectedness and the desired makespan.

        Depending on whether the negated assignment also fulfills the given formula, the backbone squares must move in the first transformation step as well, or not.
        Regardless of these movements, the configuration remains connected during the second transformation, as the schedule was legal by assumption.
    \end{claimproof}
    
    \Cref{clm:labeled-hardness-assignment-implies-schedule,clm:labeled-hardness-schedule-implies-assignment} conclude the proof.
\end{proof}

\begin{corollary}
    \label{cor:labeled-apx-hardness}
    Unless $\P=\NP$, {\textsc{Labeled Parallel Sliding Squares}} cannot be approximated within a factor better than $\nicefrac{3}{2}$ in polynomial time.
\end{corollary}

\begin{proof}
    It suffices to show that any instance $\mathcal{I}_\varphi$ created by our reduction can be solved in $3$ parallel transformations without deciding the underlying satisfiability problem.
    We proceed as follows;
    (1) solve each clause gadget and the backbone gadget in a single transformation, then
    (2) we solve the ``positive'' halves of all variable gadgets in parallel, and finally, (3) solve the ``negative'' halves.
    Due to (1), the second and third steps do not cause disconnections.
\end{proof}
%

\pagebreak%
\section{Solving scaled instances}
\label{sec:scaled-configurations}
In this section, we show that \newterm{scaled} configurations can be efficiently reconfigured into one another, allowing them to serve as canonical configurations for our main result in~\cref{sec:algorithm}.

\subparagraph*{Scale.}
In general, a configuration is \newterm{$c$-scaled} exactly if it is composed of $c \times c$ squares aligned with a corresponding $c \times c$ integer~grid.
This constraint on the input configurations closely matches the meta-module literature~\cite{aloupis.collette.damian.ea2009linear-reconfiguration,parada.sacristan.silveira2021new-meta-module,rus.vona2001crystalline-robots} that poses similar constraints by assuming that both configurations are composed of meta-modules, as outlined in~\cref{sec:introduction}.

\subparagraph*{Monotonicity and histograms.}
A configuration $C$ is \newterm{monotone} with respect to an axis exactly if for every coordinate along that axis, the modules of $C$ at this coordinate form a connected line, which we call a \newterm{bar}.
A monotone configuration is a \newterm{histogram} if all bars either share the same minimal or maximal coordinate.
We call this coordinate the \newterm{base} of $C$.
\medskip

We transfer and improve techniques from~\cite{fekete.keldenich.kosfeld.ea2023connected-coordinated} to solve scaled configurations in our model efficiently.
On a high level, the approach operates by moving modules maximally to the south-west, creating $xy$-monotone histograms as intermediate configurations.
This approach yields schedules linear in the perimeter $P_1$ and $P_2$ of the bounding boxes $B_1$ and $B_2$ of the two input configurations $C_1$ and $C_2$, respectively.

\begin{theorem}
    For {\textsc{Parallel Sliding Squares}} with two $3$-scaled configurations, an in-place schedule of at most $12(P_1+P_2)$ transformations can be computed in polynomial time.
    \label{thm:3-scaled-reconfiguration}
\end{theorem}

A crucial subroutine for our proof of~\cref{thm:3-scaled-reconfiguration} is going to be the following.

\begin{lemma}
    A $3$-scaled configuration can be translated by $k$ units in any cardinal direction by $6k$ strictly in-place transformations.
    \label{lem:scale-3-translation}
\end{lemma}
\begin{proof}
    Assume, without loss of generality, that the target direction is south.
    We divide~$C$ into maximal $3$-wide strips of modules as shown in~\cref{fig:scale-3-strip-decomposition} and move these strips to the south simultaneously, see~\cref{fig:scale-3-strip-translation}.
    \begin{figure}[htb]
        \captionsetup[figure]{justification=centering}%
        \begin{minipage}[t]{0.5\columnwidth - 0.5em}%
            \centering%
            \includegraphics[page=1]{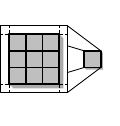}%
            \includegraphics[page=2]{scaled-translation}%
            \caption{A $3$-scaled configuration's strips.}
            \label{fig:scale-3-strip-decomposition}
        \end{minipage}%
        \hfill%
        \begin{minipage}[t]{0.5\columnwidth - 0.5em}%
            \centering%
            \transformable{\includegraphics[page=3]{scaled-translation}}%
            \ttransforms%
            {\includegraphics[page=4]{scaled-translation}}%
            \caption{Strips can be translated efficiently.}
            \label{fig:scale-3-strip-translation}
        \end{minipage}
    \end{figure}
    \begin{figure}[htb]
        \begin{subfigure}[t]{\textwidth}%
            \hfil%
            \transformable{\includegraphics[page=1]{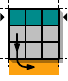}}%
            \transforms{\includegraphics[page=2]{scaled-translation-schedules}}%
            \transforms{\includegraphics[page=3]{scaled-translation-schedules}}%
            \transforms{\includegraphics[page=4]{scaled-translation-schedules}}%
            \transforms{\includegraphics[page=5]{scaled-translation-schedules}}%
            \transforms{\includegraphics[page=6]{scaled-translation-schedules}}%
            \transforms{\includegraphics[page=7]{scaled-translation-schedules}}%
            \subcaption{Moving a $3\times 3$ rectangle south.}%
        \end{subfigure}\par%
        \begin{subfigure}[t]{\textwidth}%
            \hfil%
            \transformable{\includegraphics[page=8]{scaled-translation-schedules}}%
            \transforms{\includegraphics[page=9]{scaled-translation-schedules}}%
            \transforms{\includegraphics[page=10]{scaled-translation-schedules}}%
            \transforms{\includegraphics[page=11]{scaled-translation-schedules}}%
            \transforms{\includegraphics[page=12]{scaled-translation-schedules}}%
            \transforms{\includegraphics[page=13]{scaled-translation-schedules}}%
            \transforms{\includegraphics[page=14]{scaled-translation-schedules}}%
            \subcaption{Moving a $3k\times 3$ rectangle south ($k\geq 2$).}%
        \end{subfigure}%
        \caption{Schedules to translate $3$-scaled rectangles south.}
        \label{fig:scale-3-strip-translation-schedules}%
    \end{figure}
    Using the schedules depicted in~\cref{fig:scale-3-strip-translation-schedules}, we argue that moving the strips by one unit takes exactly six transformations.
    Due to the configuration's scale and the strips' maximality, neither schedule causes collisions; they are strictly in-place for each strip.
    It remains to argue that there exists a connected backbone during each of the six transformations.
    This can be verified fully locally using~\cref{fig:scale-3-strip-translation-schedules}:
    Firstly, the modules of each strip always have a locally connected backbone which they move along.
    Furthermore, there exists an $i\in[3]$ during each transformation such that all modules that satisfy $1\not\equiv x(m)\bmod{3}$ and $i\equiv y(m)\bmod{3}$ do not move; these are marked by triangles.
    Any two adjacent strips therefore share a connected backbone, and its local existence directly implies its global existence.

    Applying this routine $k$ times yields an in-place schedule of makespan $6k$.
\end{proof}

\begin{lemma}
    Any $3$-scaled configuration $C$ with bounding box height $h$ can be transformed into a $3$-scaled histogram in at most $6(h-3)$ strictly in-place transformations.
    \label{lem:scale-3-histogram}
\end{lemma}
\begin{proof}
    Assume, without loss of generality, that the target shape has its base at $y=0$.
    We again decompose~$C$ into maximal $3$-wide strips of modules, as in~\cref{fig:scale-3-strip-decomposition}.
    To obtain the target histogram, we then iteratively move all strips that do not contain a module at $y=0$ south, using the exact schedules depicted in~\cref{fig:scale-3-strip-translation-schedules}.
    Due to~\cref{lem:scale-3-translation}, we know that adjacent strips have a common connected backbone.
    It remains to argue that the strips that already contain a module at $y=0$ have a shared, connected backbone with those that must move.
    This follows by the same line of argumentation as in the proof of~\cref{lem:scale-3-translation}:
    Every strip that moves is either adjacent to another strip that moves, or a static strip, meaning that the marked modules in~\cref{fig:scale-3-strip-translation-schedules} provide a globally connected backbone.
\end{proof}

\begin{corollary}
    Any $3$-scaled configuration $C$ can be transformed into a $3$-scaled, $xy$-monotone histogram in at most $3(P-6)$ strictly in-place transformations.
    \label{cor:scale-3-histogram-xy-monotone}
\end{corollary}
\begin{proof}
    Simply apply~\cref{lem:scale-3-histogram} twice along orthogonal axes.
\end{proof}

\begin{lemma}
    For any instance of two $3$-scaled, $xy$-monotone histograms with identical bases, we can efficiently compute an in-place schedule of at most $\max(9P_1+9P_2)$ transformations.
    \label{lem:xy-monotone-to-xy-monotone}
\end{lemma}

\begin{proof}
    Our approach generalizes an result from~\cite{fekete.keldenich.kosfeld.ea2023connected-coordinated,fekete.kramer.rieck.ea2024efficiently-reconfiguring} to reconfigure histograms in a different model, which has also been explored experimentally for the sliding squares model~\cite{wolters2024parallel-algorithms}.

    Recall that the bounding boxes $B_1$ and $B_2$ share an $xy$-minimal corner at $(0,0)$.
    We assume, without loss of generality, that the base coordinate of $C_1$ and $C_2$ with respect to either axis is $0$.
    In particular, this means that we can assume $(0,0)\in C_1\cap C_2$.

    Due to $xy$-monotonicity, ${(x_1,y_1)}\in C_1 \Rightarrow {(x_2,y_2)}\in C_1$ for all $x_2\leq x_1$ and $y_2\leq y_1$.
    This implies that we can use the schedules depicted in~\cref{fig:L-shaped-tunneling} to move a $3\times 3$ square along the slope of the configuration, taking $18$ moves regardless of distance.
    This can be performed in parallel for pairs of $3\times 3$ squares if the connecting paths as shown in~\cref{fig:L-shaped-tunneling-a} are disjoint.
    For the remainder of this proof, we use these schedules exclusively, arguing based on the non-scaled configurations instead.
    We characterize a schedule that achieves our goal in terms of such ``L''-shaped paths that can be realized in parallel in the ``real'', scaled configurations.
    \begin{figure}[htb]%
        \captionsetup[subfigure]{justification=centering}%
        \begin{subfigure}[t]{0.2\columnwidth}%
            \centering%
            \includegraphics[page=1]{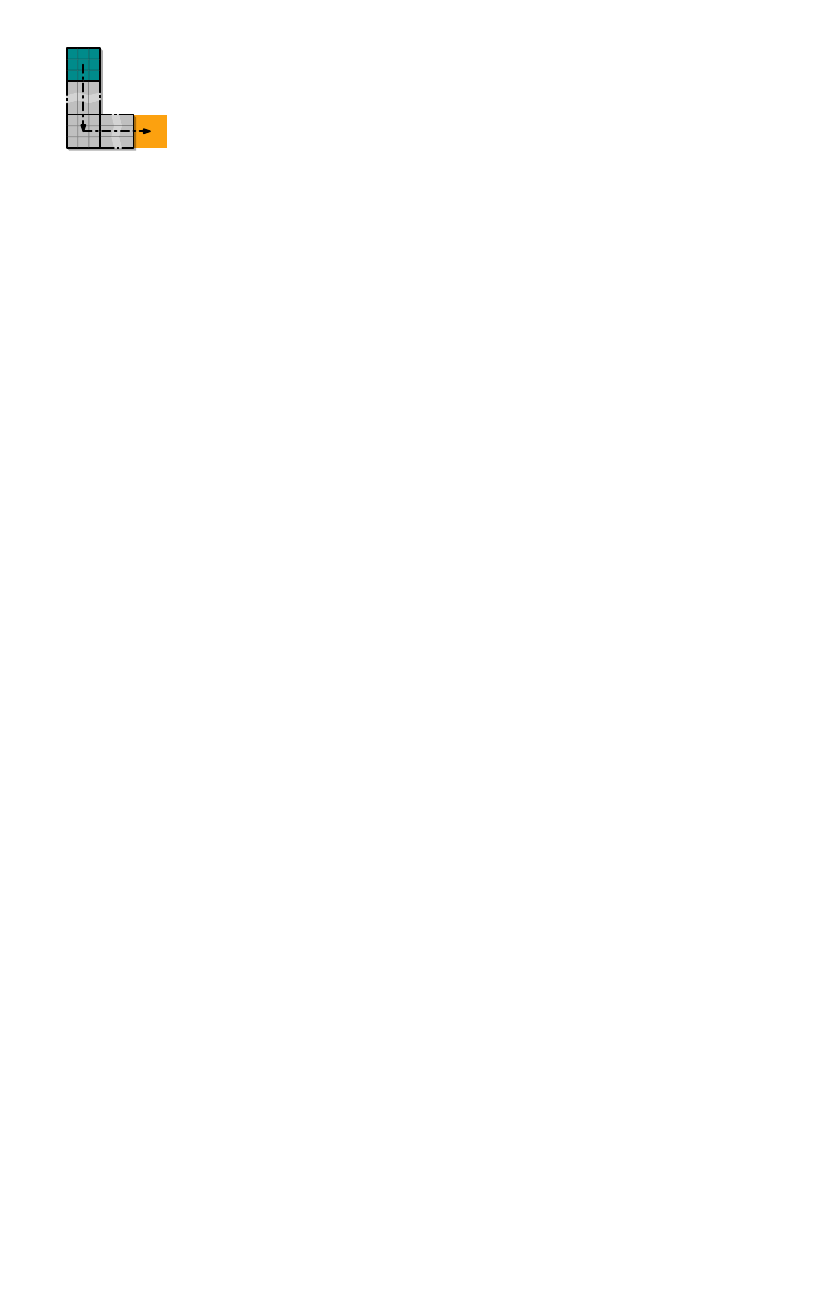}
            \caption{}
            \label{fig:L-shaped-tunneling-a}
        \end{subfigure}
        \hfill%
        \captionsetup[subfigure]{justification=raggedright}%
        \begin{subfigure}[t]{0.7\columnwidth}%
            \transformable{\includegraphics[page=2]{histogram-L-tunnel}}%
            \transforms%
            {\includegraphics[page=3]{histogram-L-tunnel}}%
            \transforms%
            {\includegraphics[page=4]{histogram-L-tunnel}}%
            \ttransforms%
            \includegraphics[page=5]{histogram-L-tunnel}%
            \caption{This pattern realizes the movement from (a) in $18$ transformations.}
        \end{subfigure}
        \caption{We can efficiently shift groups of modules along the slope of $xy$-monotone histograms.}%
        \label{fig:L-shaped-tunneling}
    \end{figure}

    We base our algorithm on two \newterm{bisectors} with south-west slopes, as shown in~\cref{fig:diagonal-bisectors}.
    A bisector $L$ with height $h$ corresponds to the set $\{(x,y)\in\mathbb{Z}^2\mid y=h-x\}$ and partitions the grid into $\east(L)=\{(x,y)\in \mathbb{Z}^2\mid y>h-x\}$, $\west(L)=\{(x,y)\in \mathbb{Z}^2\mid y<h-x\}$, and $L$.
    \begin{figure}[htb]%
        \begin{subfigure}[t]{\columnwidth/6}%
            \centering%
            \includegraphics[page=1]{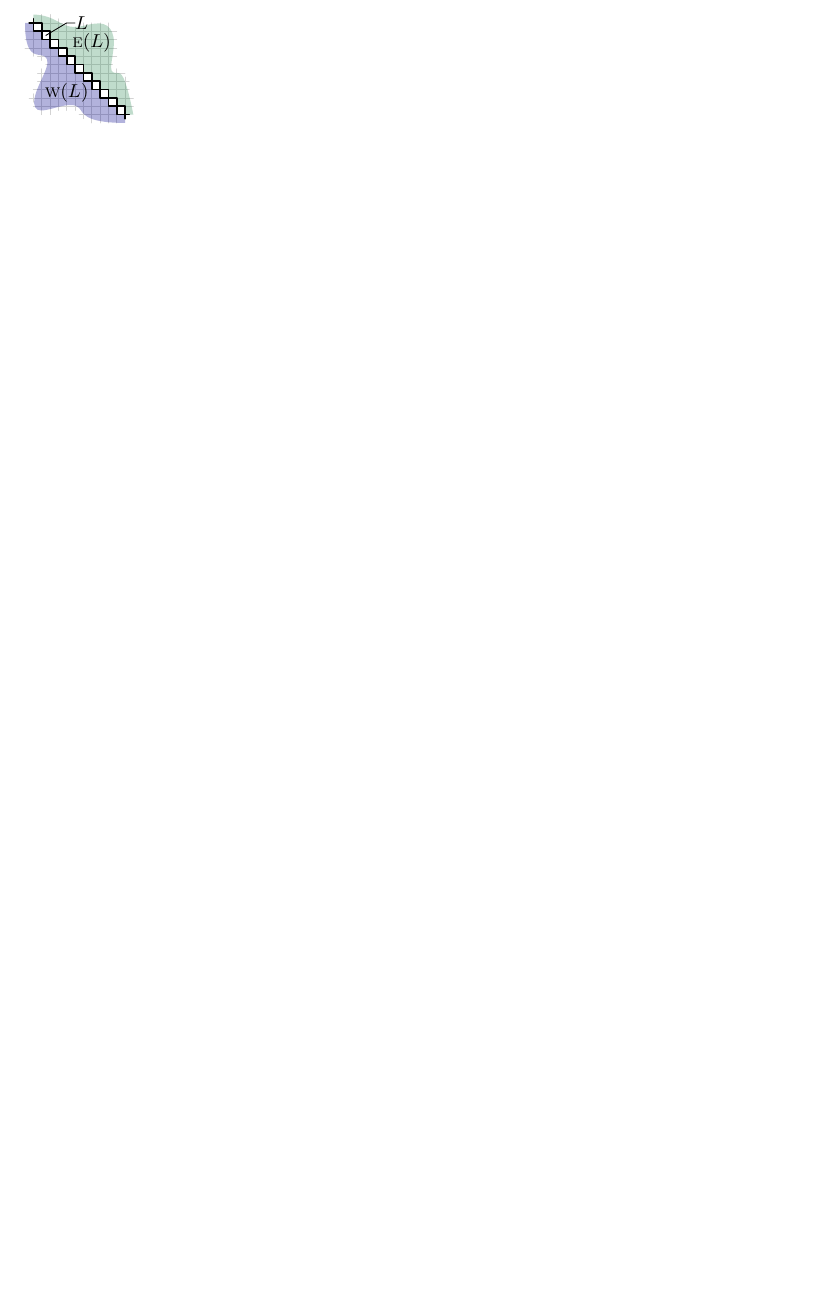}%
            \caption{A~bisector~$L$.}%
            \label{fig:diagonal-bisectors}
        \end{subfigure}%
        \hfill%
        \begin{subfigure}[t]{\columnwidth*5/12 - 1em}%
            \hfil%
            \includegraphics[page=2]{xy-monotone-balancing}%
            \hfil%
            \includegraphics[page=3]{xy-monotone-balancing}%
            \caption{Two $xy$-monotone histograms and the two bisectors $L_1$ and $L_2$.}
            \label{fig:xy-monotone-balancing-bisectors}
        \end{subfigure}%
        \hfill%
        \begin{subfigure}[t]{\columnwidth*5/12 - 1em}%
            \hfil%
            \transformable{\includegraphics[page=4]{xy-monotone-balancing}}%
            \hfil%
            \includegraphics[page=5]{xy-monotone-balancing}%
            \caption{A transformation based on $L_1$ and $L_2$ from (b) that increases the height of $L_1$.}
        \end{subfigure}
        \caption{We reconfigure according to the diagonal bisectors $L_1$ and $L_2$. Each square corresponds to a meta-module in the scaled configuration and transformations occur according to~\cref{fig:L-shaped-tunneling}.}%
        \label{fig:parallel-balancing}
    \end{figure}

    Assume now, without loss of generality, that $C_1\neq C_2$.
    Let $L_1$ be the highest bisector such that for every $m_1\in \west(L_1)$, ${m_1\in C_2\Rightarrow m_1\in C_1}$.
    By definition, there exists at least one cell in $L_1$ that is occupied in $C_2$, but not in $C_1$.
    Further, let $L_2$ be the highest bisector such that there exists an $m_2\in L_2$ with $m_2\in C_1\setminus C_2$.

    Due to~\cite[Claim 4]{fekete.keldenich.kosfeld.ea2023connected-coordinated}, we compute a maximum matching in $(L_1\cap(C_2\setminus C_1))\times(L_2\cap(C_1\setminus C_2))$ that has a plane embedding using ``L''-shaped paths.
    In any such maximum matching, either all unoccupied cells in $L_1$, or all occupied cells in $L_2$ are matched.
    We use the schedules from \cref{fig:L-shaped-tunneling} to move each corresponding meta-module in the scaled configuration to its matched center cell.
    As a result, either (a) every previously unoccupied target cell in $L_1$ is now occupied, or (b) every excess module in $L_2$ has been removed.
    Updating both bisectors afterward, it follows that either $L_1$ has moved north, or $L_2$ has moved south.

    After at most $\max(\nicefrac{P_1}{2},\nicefrac{P_2}{2})$ iterations of this process, either $C_1\subset\east(L_2)$, or $C_1\subset\west(L_1)$.
    In either case, it follows that $C_1 = C_2$.
    Each iteration takes exactly $18$ transformations, so the entire schedule has makespan at most
    ${18\max(\nicefrac{P_1}{2},\nicefrac{P_2}{2})\leq\max(9P_1+9P_2)}$.
\end{proof}
We further note the following property of schedules derived by the above method.
\begin{observation}
    The schedules constructed according to~\cref{lem:xy-monotone-to-xy-monotone} never move a module located at $(0,0)$ if this module is part of both $C_1$ and $C_2$.
\end{observation}

\cref{thm:3-scaled-reconfiguration} now immediately follows from \cref{cor:scale-3-histogram-xy-monotone,lem:xy-monotone-to-xy-monotone}.
    \section{A worst-case optimal algorithm}
\label{sec:algorithm}
In this section, we introduce a three-phase algorithm for efficiently reconfiguring two given configurations $C_1, C_2$ into one another.
Our approach computes reconfiguration schedules linear in the perimeter $P_1$ and $P_2$ of the bounding boxes $B_1$ and $B_2$ of $C_1$ and $C_2$, respectively.

\begin{restatable}{theorem}{theoremAlgorithm}
	\label{thm:algorithm}
	For any instance $\mathcal{I}$ of {\textsc{Parallel Sliding Squares}}, we can compute a weakly in-place schedule of $\mathcal{O}(P_1+P_2)$ transformations in polynomial time.
\end{restatable}
\pagebreak%

Our algorithm consists of the three phases depicted in~\cref{fig:algorithm-overview}:
In \phaseref{Phase~(I)}, we identify a subconfiguration used as a ``backbone'' (similar to~\cite{hurtado.molina.ramaswami.ea2015distributed-reconfiguration}), and gather $\Theta(P_1)$ many modules around a piece of this backbone, thereby enhancing its connectivity (as in~\cite{akitaya.demaine.korman.ea2022compacting-squares}).
We use the flexibility of this piece to construct a sweep-line structure out of meta-modules in \phaseref{Phase~(II)}, that is then used in \phaseref{Phase~(III)} to efficiently compact and transform the remaining modules into grid-aligned $3\times 3$ squares, creating a $3$-scaled configuration.
Due to~\cref{thm:3-scaled-reconfiguration}, these can be reconfigured efficiently and will serve as a canonical intermediate configuration.
To~reach the respective target configuration, we then simply apply \phaseref{Phases~(I-III)} in reverse.
Although several of our techniques are inspired by previous work, they differ substantially in our context of parallel reconfiguration and considerable changes were needed.

\begin{figure}[htb]
	\begin{subfigure}[t]{\textwidth/3 - 1em}
		\centering
		\includegraphics[page=1,scale=2]{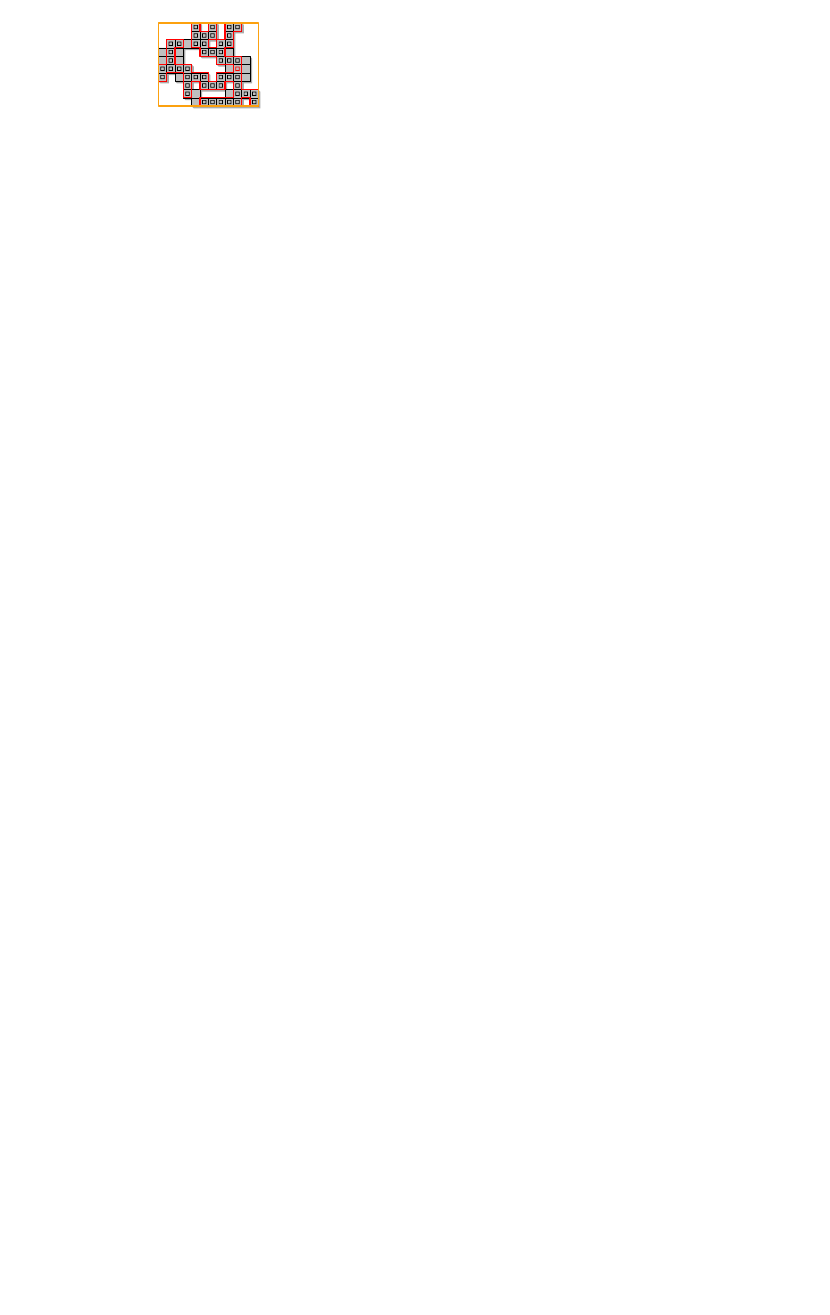}%
		\subcaption{Initial configuration and a rooted skeleton.}
		\label{fig:overview-a}
	\end{subfigure}%
	\hfill%
	\begin{subfigure}[t]{\textwidth/3 - 1em}
		\centering
		\includegraphics[page=2,scale=2]{overview}%
		\subcaption{Constructing an exoskeleton.}
		\label{fig:overview-b}
	\end{subfigure}%
	\hfill%
	\begin{subfigure}[t]{\textwidth/3}
		\centering
		\includegraphics[page=3,scale=2]{overview}%
		\subcaption{Creating a sweep line.}
		\label{fig:overview-c}
	\end{subfigure}%
	\par\medskip%
	\hfil
	\begin{subfigure}[t]{\textwidth/2 - 0.5em}
		\centering
		\includegraphics[page=4,scale=2]{overview}%
		\subcaption{After sweeping into a histogram.}
		\label{fig:overview-d}
	\end{subfigure}%
	\begin{subfigure}[t]{\textwidth/2 - 0.5em}
		\centering
		\includegraphics[page=5,scale=2]{overview}%
		\subcaption{An $xy$-monotone, $3$-scaled configuration.}
		\label{fig:overview-e}
	\end{subfigure}%
	\caption{The high-level overview of our approach. (a) Initial configuration and (b), (c) and (e) are the result of  \phaseref{Phases (I--III)}, respectively. (d) shows an intermediate configuration in \phaseref{Phase~(III)}. The original bounding box is shown in yellow. The coloring of the squares is explained later in the following sections.}
	\label{fig:algorithm-overview}
\end{figure}

\begin{restatable}{lemma}{lemAlgorithm}
	\label{lem:worst-case}
	The bounds achieved in~\cref{thm:algorithm} are asymptotically worst-case optimal.
\end{restatable}
\begin{proof}
	Assume that $n$ is even and consider two configurations $C_1, C_2$ contained in an $\nicefrac{n}{2}\times \nicefrac{n}{2}$ bounding box $B$.
	Define~$C_1$ to contain the $n - 1$ modules adjacent to the left and bottom edges of $B$ and a module at~$(1,1)$.
	Define $C_2$ to contain the $n - 1$ modules adjacent to the top and right edges of $B$ and a module at $(\nicefrac{n}{2}-1,\,\nicefrac{n}{2}-1)$.
	Then, the minimum bottleneck matching between full cells in $C_1$ and $C_2$ has bottleneck $\Omega(n)$, implying that a module must make  $\Omega(n)$ moves.
	The makespan is then $\Omega(n)=\Omega(P)$, where $P$ is the perimeter of $B$.
\end{proof}
    \subsection{Phase~(I): Gathering squares}
\label{subsec:gather}
In \phaseref{Phases~(I)} and \phaseref{(II)}, we use an underlying connected substructure (called \newterm{skeleton}) of the initial configuration to guide the reconfiguration. 
Intuitively (formal definitions below), this tree-like skeleton functions as a backbone around which we move modules toward a ``root'' module, making a subtree ``thick'' (filling cells in the neighborhood of the skeleton).
This ``thick'' subskeleton is more manipulable, making it easier to mold it into a sweep line in~\phaseref{Phase~(II)}.
The authors of~\cite{hurtado.molina.ramaswami.ea2015distributed-reconfiguration} use a similar strategy of moving modules around a tree.
This type of approach becomes complicated when the tree creates bottlenecks where collisions might happen. 
In~\cite{hurtado.molina.ramaswami.ea2015distributed-reconfiguration} this is solved by strengthening the model and allowing modules to ``squeeze through'' bottlenecks.
We achieve a stronger result in the classic (unmodified) sliding model via careful definition of the tree-like structure and movement~schedules.

\subparagraph*{Skeleton.}
A \newterm{skeleton} of a configuration $C$ is a valid subconfiguration $S$ such that
$C\subseteq N[S]$ (every module of $C$ is either contained in $S$ or edge-adjacent to a module of~$S$); and cycles in the dual graph of $S$ are pairwise disjoint with length at most 4.
We~call a module in $S$ (resp., not in $S$) a skeleton (resp., nonskeleton) module.
In~\cref{fig:overview-a}, skeleton modules are highlighted with an internal square and the perimeter of the skeleton is shown in red.
Note that any set of nonskeleton modules is free.

On a high level, we can compute a skeleton $S$ of a given configuration $C$ as follows.
Think of $S$ as a subset of $C$, initialized as the empty set.
Modules of $C$ are then algorithmically added to $S$ until it satisfies the definition of a skeleton.
First, all modules with even \mbox{$x$-coordinate} are added to $S$, followed by all modules with odd $x$-coordinate with no east and west neighbors.
Next modules are added to $S$ until it is a connected subconfiguration of $C$ (see magenta squares in~\cref{fig:overview-a}).
This might introduce large cycles (greater than length~4) to~$S$.
These~cycles are broken via removal of modules from $S$ or exchanging membership in $S$ between adjacent modules.
We~show that careful execution of these steps will build $S$ into a skeleton of $C$.

A concrete discussion of the construction of a skeleton now follows.
Given a collection of modules~$M$. Let $M[i]$ be the subset of~$M$ of all modules with $x$ coordinate $i$, likewise let~$M[i,j]$, where $j > i$, be the set of all modules with $x$ coordinate in the interval~$[i,j]$.
We~can compute a skeleton~$S$ of a configuration as follows in \Cref{alg:skeleton}.

\begin{algorithm}[ht]
	\caption{Given a configuration $C$, compute a skeleton $S$ of $C$.}
	\begin{algorithmic}
		\STATE $i \leftarrow 0$
		\STATE $S \leftarrow \{\emptyset\}$
		\STATE $V \leftarrow \{\emptyset\}$
		\WHILE{$C[i] \neq \{\emptyset\}$} \label{ln:skeleton1}
		\STATE Add each maximally connected component of $C[i]$ to $V$
		\STATE $i \leftarrow i + 2$
		\ENDWHILE
		\STATE $V \leftarrow V \cup \{ C \setminus V \text{ and } V \text{'s 1-neighborhood}\}$ \label{ln:skeleton2}
		\STATE $S \leftarrow V$
		\STATE $i \leftarrow 0$
		\WHILE { $V[i] \neq \{\emptyset\}$}
		\STATE Create a graph $G$ with a vertex for each component in $V[i, i+2]$.
		\FOR {Each pair of vertices $v_1, v_2 \in G$}
		\IF { $v_1$ $v_2$ are connected by a path in  $C[i] \cup C[i+1] \cup C[i+2]$ containing no other element of $V[i,i+1,i+2]$}
		\STATE Add an edge $\{v_1, v_2\}$ to $G$ with a weight equal to the number of modules in the shortest such path.
		\ENDIF
		\ENDFOR
		\STATE Compute a minimum spanning forest $F$ of $G$
		\STATE For each edge in $E(F)$ add the corresponding modules of $C$ to $S$.
		\WHILE{$S$ has a cycle $L$ of length greater than $4$} \label{ln:cycle-break}
		\STATE Use~\cref{lem:cycle-break} to break $L$
		\ENDWHILE
		\STATE $i \leftarrow i + 2$
		\ENDWHILE
		\WHILE{$S$ has $4$-cycles that are not disjoint}
		\STATE Apply~\cref{lem:disjoint-cycles} to remove one of the cycles.
		\ENDWHILE
		\RETURN{$S$}
	\end{algorithmic}
	\label{alg:skeleton}
\end{algorithm}

We now prove the correctness of \cref{alg:skeleton} via a sequence of lemmas.

\begin{lemma}
	\label{lem:cycle-break}
	During \cref{alg:skeleton} at Line \ref{ln:cycle-break} if the dual graph of $S$ has a cycle $L$ of length greater than 4, then $S$ can be modified so that every module of $C$ is still a member of $S$ or adjacent to $S$ and $L$ is broken.
\end{lemma}
\begin{proof}
	First, note if $G$ has a 4-cycle $f$ and at least 2 modules of this cycle are also elements of $L$, then one element of $f$ can be removed from $S$.
	At least two modules of $f$, say $a,b$ are connected by $L$.
	We may be able to remove one of them from $S$ breaking $L$.
	However if $a$ or~$b$ have a neighboring module that is an element of $S$ but not of $f$ or $L$,
	they might not be removable without disconnecting $S$.
	However, if $a$ and $b$ have this many neighbors one of the other two modules of $f$ is safe to remove.
	This is because any of their potential neighbors are either elements of $f$ or adjacent to a neighbor of $a$ or $b$.
	Therefore we assume for any two east-west adjacent modules of $L$, they do not both simultaneously have a south neighbor that is an element of $S$,
	or likewise, neither has a north neighbor which is an element of $S$.
	The existence of these simultaneous neighbors would form a $4$-cycle.

	We now prove that, in the remaining cases, there is always a module that can be removed from $L$,
	via case analysis.
	Suppose there is a module $r \in L$ that has a north neighbor $r_n$ that is not an element of~$S$.
	Then, by the construction of $S$ in~\cref{alg:skeleton}, either $r_n$ has a west or east neighbor that is an element of~$S$ and so $r_n$ was not added to $V$ initially.
	Or~$r_n$ was an element of~$S$ but was later removed from $S$ to break a cycle.
	However, we will maintain the invariant that when we remove a module from $S$ to break a cycle it always has a north, east, or west neighbor that is an element of~$S$.
	Therefore, if $r$ is removed from $S$, $r_n$ will still be adjacent to an element of $S$ through a north, east, or west neighbor.
	Hence, if there is a module $r \in L$ such that $r$ has no south neighbor, and if it has a north neighbor $r_n$, and $r_n$ is not an element of~$S$, we can remove $r$ from $S$, breaking $L$.

	Suppose every module of $L$ has a north or south neighbor.
	Observe that this is possible, as shown in~\cref{fig:nscycle}.
	As $L$ is a cycle in a geometric grid graph, there must be a sequence of three modules $a,b,c \in L$ that form a ``corner'', an
	L-shaped arrangement where $a$ and $b$ are north/south neighbors and $b,c$ are east/west neighbors.
	Assume $b$ is the north neighbor of $a$ and the east neighbor of $c$.
	That is, $a,b,c$ form a ``top-left'' corner of $L$, with $b$ in the middle.
	If $b$ has no neighbors other than $a$ and $c$, then we can remove it from $C$.

	Therefore suppose $b$ has a north neighbor $b_n$. If $b_n \in S$, then the only neighbor of $c$ that can be an element of $S$
	is its east neighbor, as otherwise a 4-cycle would be formed by $b_n, b, c, c_n$ or $b,c,c_s,a$
	(where $c_n$ and $c_s$ are $c$'s north and south neighbors, respectively).
	Therefore, if $b_n \in S$, we can remove $c$ from $L$ as its east neighbor must be an element of $L$, and if it does
	have north or south neighbors, they are adjacent to $b_n$ or $a$.
	Similarly if $b$ has a west neighbor, $b_w$, which is an element of $S$, $a$ can be removed from $S$.
	By the arguments made at the start of this proof, if $b_n \notin S$ then $b_n$ has a west or east neighbor that is an element of $S$.
	So if $b$ is removed from $S$, $b_n$ will still be adjacent to a module of $S$.
	As a consequence if $b$ has no west neighbor $b_w$ and $b_n \notin S$ we can remove $b$ from $S$.

	Finally, consider the case that $b$ has a west neighbor $b_w$. As argued earlier, if $b_w \in S$ then $a$ can be removed
	from $S$. So suppose $b$ has a west neighbor $b_w \notin S$. We now have two cases:
	\begin{description}
		\item[$b_w$ has a neighbor $x \neq b$, where $x \in S$:] We can remove $b$ from $S$.
		\item[$a$ has a west neighbor $a_w$:] If $a_w \in S$, then $b$ can be removed from $S$ as $b_w$ is adjacent to $a_w$.
		If $a_w \notin S$ but has a west neighbor that is an element of $S$, we can add $a_w$ to $S$ while removing $a$ and $b$ from $S$.
		We remove $a$ and $b$ instead of just $b$ as adding $a_w$ could form a new cycle with its west neighbor and $a$. If $a_w$
		does not have a west neighbor that is an element of $S$ then we can add $a_w$ to $s$ and remove $b$ from $S$.  \qedhere
	\end{description}
\end{proof}

\begin{figure}[htb]
	\begin{minipage}[t]{0.5\columnwidth - 0.5em}%
		\centering%
		\includegraphics[page=1]{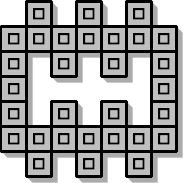}
		\caption{A large cycle where every module has a north/south neighbor in $S$.}
		\label{fig:nscycle}
	\end{minipage}%
	\hfill%
	\begin{minipage}[t]{0.5\columnwidth - 0.5em}%
		\centering%
		\includegraphics[page=2]{skeleton-cycles}%
		\caption{Two intersecting $4$-cycles. The top left cycle is labeled as in the proof of \cref{lem:disjoint-cycles}.}
		\label{fig:disjoint4cycle}
	\end{minipage}
\end{figure}
\begin{lemma}
	\label{lem:disjoint-cycles}
	If the dual graph of a skeleton $S$ has two intersecting $4$-cycles then one of the cycles can be broken.
\end{lemma}
\begin{proof}
	Let $c_1$ and $c_2$ be two $4$-cycles of the dual graph of $S$ such that $c_1 \cap c_2 \neq \{\emptyset\}$.
	We~assume the intersection of $c_1$ and $c_2$ is exactly one module $x$. It is trivial that in a grid graph, the intersection of two
	distinct $4$ cycles can only be 1 or 2 elements else $c_1$ and $c_2$ are the same cycle.
	If the intersection of $c_1$ and $c_2$ is precisely two elements, then $c_1 \cup c_2$
	would form a cycle of length $6$, in this case we could apply~\cref{lem:cycle-break}.

	Without loss of generality assume that the shared module, $x$, is the bottom right module of $c_1$ and the top left module of $c_2$.
	Let $\chi$ represent the union of these two cycles, $\{c_1 \cup c_2\}$.

	First, we handle a ``base case''. Suppose for every module $m \in C$ the only neighbors of~$m$ that are elements of $S$ are
	also elements of $\chi$. Note, that this implies that the entirety of the skeleton is $\chi$ as if there were some
	module of $S$ that is not part of these cycles, by assumption, it is not adjacent to any module of $\chi$, and $S$ would be disconnected.
	In this case, if there is a module of $\chi$ other than $x$ with no neighbor outside of $S$, we can remove it from $S$ and break one of the cycles.
	Otherwise every module besides $x$ has a neighbor not in~$S$, an example of one such configuration is shown in \cref{fig:disjoint4cycle}.
	In this case, any module of $\chi$ can be removed from $S$ and be replaced with a vertex adjacent module not in $S$.

	Now, we assume some module of $\chi$ has a neighbor $y$ where $y \in S$ but $y \notin \chi$.
	Let the top right module of $c_1$ be $r$, the bottom left be $l$ and the top left module be $m$.
	Without loss of generality assume $y$ is adjacent to some module of $c_1$, some short casework follows.

	Suppose $y$ is adjacent to $m$.  If $r$ has a north neighbor it must be the east neighbor of $y$, and if $r$ has
	an east neighbor, it is adjacent to an element of $c_2$, by our assumed arrangement. Therefore we can remove $r$ from $S$.

	Suppose there is no module of $S$ adjacent to $m$. If $y$ is adjacent to $r$ then we then try to remove $m$ from $S$.
	$m$ can have two neighbors outside of $c_1$. A north neighbor and a west neighbor $c_w$. $m$ having more neighbors
	only makes removing $m$ while maintaining skeleton properties harder, so assume $m$ has both neighbors.
	$y$ must be the north neighbor of $r$, otherwise it would form at least a 6 cycle with $c_2$.
	Therefore $c_n$ is adjacent to $y$. So of $m$'s neighbors, removing $m$ from $S$ could only at most result in $c_w$ being the only vertex not adjacent to any vertex in $S$.
	If $l$ has a west neighbor $l_w$ and $l_w \in S$, then we can safely remove $m$ from $S$.
	If $l_w \notin S$, then we can add $l_w$ to $S$ and remove $m$. It may be possible that adding $l_w$ to $S$ creates
	a large cycle in $S$. However, this cycle would necessarily use $x$ and one other module of $c_2$, and as argued in the
	proof of~\cref{lem:cycle-break} if a large cycle involves at least two modules of a 4-cycle, that 4-cycle can be broken.
	So, we add $l_w$ and remove $m$, and if a large cycle is formed, we remove a module of $c_2$ from $S$.

	If $l$ has no west neighbor then then we remove $l$ instead of $m$. If $l$ has a south neighbor,
	that module must be adjacent to a module of $c_2$ and is therefore adjacent to another module of $S$, other than $l$.
\end{proof}

\begin{lemma}
	\label{lem:valid-skeleton}
	The subconfiguration $S$ created by~\cref{alg:skeleton} is a skeleton of $C$.
\end{lemma}
\begin{proof}
	There are three things we need to prove about $S$.
	One, it is a valid subconfiguration of $C$.
	Two, every module of $C$ is an element of $S$ or adjacent to $S$.
	Three, $S$ has no cycle larger than a $4$-cycle, and all $4$-cycles are disjoint.

	\begin{description}
		\item[$S$ is a valid subconfiguration of $C$:] First, ignore the modifications made to $S$ in the while loop on line~\ref{ln:cycle-break}.
		We construct $S$ by taking three sequential subsets of $V$, $V[i]$, $V[i+1]$, and $V[i+2]$ and connecting them via a spanning forest.
		We then connect $V[i+2]$, $V[i+3]$, and $V[i+4]$. Due to the overlap of these $V$'s by the end of our iteration we have connected every
		element of $V$, as long as $C$ is a connected configuration. So $S$ is a connected configuration.
		The modifications on line~\ref{ln:cycle-break} cannot disconnect $S$, as they take a cycle $L$ of $S$, and remove exactly one element of $L$ from $S$.

		\medskip
		\item[Every element of $C$ is an element of $S$ or adjacent to an element of $S$:] By our formulation of $V$, initially
		every element of $C$ is either an element of $V$ or adjacent to $V$. So by initializing $S$ to $V$, $S$ immediately meets
		this condition. Hence, the only way to break this condition is in the while loop on line~\ref{ln:cycle-break} where
		modules are removed from $S$. However, by~\cref{lem:cycle-break} these modifications maintain the desired adjacencies.

		\medskip
		\item[$S$ has no cycle larger than a $4$-cycle and all $4$-cycles are disjoint:]
		We remark that it may be impossible to remove all $4$-cycles.
		For example, near the top of~\cref{fig:overview-a} there is an example of a $4$-cycle
		must be included in $S$ or else some module of $C$ will not be adjacent to $S$.
		The fact that there are no large cycles in $S$ and all $4$-cycles are disjoint follows from~\cref{lem:cycle-break,lem:disjoint-cycles}.
		\cref{alg:skeleton} uses these lemmas inside while loops to remove any unwanted cycles until no more exist. \qedhere
	\end{description}
\end{proof}

Consider the dual graph is a skeleton $S$.
Contracting every cycle in the dual of $S$ to a single vertex renders a max-degree-4 tree where each node is either a cell or a 4-cycle.
We refer to the nodes of this tree as \newterm{nodes of $S$} and abuse notation by referring to $S$ as a tree.
Let $r$ be the \newterm{root} node of $S$, arbitrarily chosen.
For every nonskeleton module, we assign an adjacent module of $S$ as its \newterm{support}.
We define the \newterm{subskeleton rooted at~$c$} (denoted~$S_c$) as the subconfiguration containing $c$ and its descendant nodes.
The set $S_c^*$ contains the modules in $S_c$ and the nonskeleton modules supported by modules in $S_c$.
The \newterm{weight} of a subskeleton $S_c$ is defined as $|S_c^*|$.
By definition, $S_r^*=C$ and $|S_r^*|=n$.

\begin{restatable}{lemma}{lemHeavyNode}
    \label{lem:heavy-node}
    Given a rooted skeleton $S$ and an integer $1<w\le n$, there exists a node $c$ of~$S$ such that $w\le |S_c^*|\le 3w$.
\end{restatable}

\begin{proof}
	Let $r$ be the root of $S$. We can assume that $r$ is initially a leaf of $S$, so that when we recurse on subtrees the maximum number of children of a node is 3.
	By induction on $n$.
	\begin{description}
		\item[Base case:] When $1/3n\le w$, $S_r$ satisfies the requirements of the lemma.
		\item[Inductive step:] If a child $c$ of $r$ satisfies the requirements, we are done.
		If $|S_c^*|< w$ for every child $c$ of $r$,
		then  $|S_r^*|\le 3w-2$ since the maximum number of children is 3, thus we are in the base case.
		Else, $|S_c^*|> 3w$ for some child $c$ of $r$ and, by inductive hypothesis, the claim is true for a subskeleton of $S_c$.\qedhere
	\end{description}
\end{proof}

Using this, we locate a module $h \in S$ (highlighted with a red marker in~\cref{fig:overview-a,fig:overview-b}) such that $|S_h^*|\in \Theta(P)$.
Afterwards, we ``thicken'' $S_h^*$ into a structure called the~\newterm{exoskeleton}.

\subparagraph*{Exoskeleton.}
An exoskeleton is made from three types of modules: \newterm{core}, \newterm{shell}, and \newterm{tail} modules; depicted respectively in dark gray, purple and pink in the figures.
Recall that we wish to make the neighborhood of $S_h$ is full.
The core modules occupy the positions originally occupied by the skeleton $S_h$, the shell modules are the ones ``coating'' the skeleton, and the tail modules are ``leftover'' modules in the neighborhood of a leaf.
We recursively thicken $S_h$ by applying this process to its smaller subtrees, effectively converting tail or core modules at the leaves into a shell module near the root.
If the core modules cover the entire subtree of $S_h$, it may take a linear number of transformations to move a module from a leaf to the root.
Instead, we allow empty positions inside the exoskeleton which permits us to ``teleport'' a module from a leaf to the root in $\mathcal{O}(1)$ transformations.
Although some of the positions are empty, the structure remains connected due to its shell and the fact that these empty positions are well separated (in tree metric along $S_h$, see Property~\ref{itm:core-4} below).

Formally, we define an \newterm{exoskeleton} $X_c$ as follows.
The \newterm{core} $\overline{X_c}$ of $X_c$ are positions in the lattice (not necessarily full) that form a subtree of $S_c$ containing~$c$ (since we recursively ``shave off'' leaves to make material for the shell).
Note that $\overline{X_c}$, unlike $S_c$, is not a subgraph of the configuration's dual graph because of the empty positions.
In our figures, we highlight core cells with an internal circle.
Let $\mathcal{L}$ be the set of positions corresponding to the leaves of~$\overline{X_c}$.
We call $N^*(\overline{X_c}\setminus\mathcal{L})\setminus\mathcal{L}$ the \newterm{shell} of $X_c$ (recall $N^*(S)$ is the open neighborhood under vertex-adjacency).
The following must hold:
\begin{enumerate}
    \item \label{itm:core-1} All modules in $X_c$ are in the neighborhood of its core. ($|\overline{X_c}|\ge 2$, and $X_c\subset N^*\left[\overline{X_c}\right]$.)
    \item \label{itm:core-2} The leaves are occupied. ($\mathcal{L}\subset X_c$.)
    \item \label{itm:core-3} All cells in the shell are occupied. ($N^*(\overline{X_c}\setminus\mathcal{L})\setminus\mathcal{L}\subset X_c$.)
    \item \label{itm:core-4} The depth (in $\overline{X_c}$) of every empty cell is congruent to $k$ (mod 4) for a fixed $k\in\{0,1,2,3\}$.
\end{enumerate}

\begin{figure}[ht]
    \captionsetup[subfigure]{justification=centering}%
    \begin{subfigure}[t]{0.5\columnwidth}%
        \centering%
        \includegraphics[page=1]{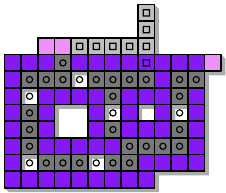}%
        \subcaption{}%
    \end{subfigure}%
    \begin{subfigure}[t]{0.5\columnwidth}%
        \centering%
        \includegraphics[page=2]{exoskeleton-ipe}%
        \subcaption{}%
    \end{subfigure}%
    \caption{Visualizing an exoskeleton. (a) Skeleton and core cells are highlighted with an inner square and circle, respectively. (b) Highlights the core; the root is highlighted with a red circle.}
    \label{fig:exoskeleton}
\end{figure}

The modules that are in $N^*(\ell)$, for a leaf $\ell \in\mathcal{L}$, and not in the shell or core of $X_c$ are called the \newterm{tail} modules of $\ell$.
By definition, a module in~$X_c$ is either in the shell, in the core, or is a tail module.
Shell modules protect the connectivity of $X_c$, allowing us to place empty positions in the core wherever necessary to expand the exoskeleton constant time.

We show that $S_h$ can be reconfigured into an exoskeleton $X_h$ in $\mathcal{O}(P)$ transformations, thereby inductively establishing~\cref{lem:gather}.
To this end, we first prove~\cref{lem:base-case1}, demonstrating how small skeletons (lighter than 9) are turned into exoskeletons.
In doing so, we obtain~\cref{cor:base-case}, which settles the base case of our induction.
\Cref{lem:base-case1} will also be pivotal in the inductive step.
All that this lemma does is take a subskeleton of weight $\leq 9$ and transform it into a $3\times 3$ square, a minimal instance of an exoskeleton $(|X_c| = 2)$.

At first glance, \cref{lem:base-case1} (below) might seem overly complicated, since without any obstacle, modules can freely move along the perimeter of $S$ (see \cref{fig:base-case-small-a}).
However, nonlocal pieces of the skeleton (connected subsets $S_1$ and $S_2$ of $S$ are \newterm{local} if $S \cap N^*(S_1\cup S_2)$ is connected) previously converted to exoskeletons can interfere when processing later subskeletons.
\Cref{fig:base-case-small-b,fig:base-case-small-c} show a configuration before and after two nonlocal pieces of the skeleton were converted into exoskeletons.
Modules in the working subskeleton~$S_c$ might be part of the shell of another exoskeleton, making their movement dangerous (potentially disconnecting the nonlocal exoskeletons).
We try to leave such modules in their place, changing their membership from $S_c^*$ to the shell of the exoskeleton.

\begin{restatable}{lemma}{lemBaseCase}
    \label{lem:base-case1}
    Let $C$ be a configuration with skeleton $S$ in which potentially some of its subskeletons were reconfigured into exoskeletons.
    Let $c$ be a skeleton node, with parent node~$d$, such that $S_c$ is not yet part of an exoskeleton and $|S_c^*|\le 9$.
    Let $M$ be the set of modules in~$S_c^*$ that are not contained in exoskeletons.
    Then, in $\mathcal{O}(1)$ transformations, $M$ can be reconfigured so~that:
    \begin{itemize}
        \item if $\deg(d)=2$, either $N^*[d]$ is full or $M\cap(N^*[c]\setminus N^*[d])$ is empty (i.e., the modules in~$M$ are all contained in the neighborhood of $d$ if this region is not full);
        \item if $\deg(d)>2$, either $N^*[d]\cap N^*[c]$ is full or $M\cap(N^*[c]\setminus N^*[d])$ is empty (i.e., the modules in $M$ are all contained in the intersections of the neighborhood of $d$ and $c$ if this region is not full).
    \end{itemize}
\end{restatable}

\begin{proof}
	We show this by induction. The base case is when the conditions are already true.
	We update $S_c$ by reassigning the modules $S_c^*$ contained in other exoskeletons to their respective exoskeletons.
	If all modules of $M$ are in $N^*[c]$, we can delete $c$ from the skeleton, since all such modules can be supported by either $c$'s sibling (if it exists), $d$, or $d$'s parent.
	For the inductive step, we focus on the deepest node $c'$ of $S_c$ that does not satisfy the conditions, let~$d'$ be its parent.
	We first assume that $S_{c'}^*\subseteq N^*[{c'}]$ and that $\deg(d')=2$.
	If the neighborhoods of $c'$ and $d'$ do not contain nonlocal pieces of the skeleton or exoskeletons, then the modules in $N^*[c']\setminus N^*[d']$ can all freely move on the perimeter of $S_{c'}$ (\cref{fig:base-case-small-a}). We then move them to $N^*[d']$.
	Else, if the perimeter of $S_{c'}$ is obstructed by nonlocal parts, we can move the skeleton modules of $S_{c'}$ along the perimeter of the nonlocal parts.
	This will not break connectivity because we can temporarily ``park'' modules of $M$ in the perimeter of modules not in $M$.
	If $M$ is not a cut set, an empty position in $N^*[d']$ can easily be filled within 2 parallel moves.
	An interesting case occurs when $M$ is a cut and we must move the parent~$m$ of $d'$ to fill an empty position (\cref{fig:base-case-small-b}).
	In this situation, $m$ must have degree 3, or else $M$ is not a cut set, and $N^*[d']\cap N^*[c']$ must be full, or else we can fill a position in there with a simpler sequence of moves.
	Since $N^*[d']\cap N^*[c']$ is full and connected to a nonlocal part of the configuration, moving $m$ does not break connectivity. We can then move the module in $d'$ to $m$, restoring connectivity with $m$'s parent.
	This creates an empty position in $N^*[d']\cap N^*[c']$ that can be filled using one of the previous techniques.
	So far we avoided moving any module in a nonlocal structure.
	However, this may be unavoidable if the skeleton modules in $S_{c'}$ are part of the shells of exoskeletons (\cref{fig:base-case-small-c}).
	By the definition of a skeleton, this is only possible if the core of an exoskeleton intersects  $N^*[d']\cup N^*[c']$ at a degree-2 bend.
	If this core cell is full, moving skeleton modules of $S_{c'}$ do not affect connectivity of their exoskeletons.
	Else, such a cell is empty, but its neighborhood must be full by Properties~\ref{itm:core-3} and \ref{itm:core-4} of exoskeletons.
	We arrive at the same conclusion that moving skeleton modules of $S_{c'}$ is safe.
	We now address the case when $\deg(d')>2$ and $S_{c'}^*\subseteq N^*[c']$ (\cref{fig:base-case-small-d}).
	We just need to fill $N^*[d']\cap N^*[c']$.
	Using techniques from the previous cases, we can always find paths from a module in $N^*[c']\setminus N^*[d']$ to an empty position in $N^*[d']\cap N^*[c']$ while avoiding $d$ and its parent (which is either outside or in the boundary of $N^*[d']\cap N^*[c']$), and avoiding cells in shells of exoskeletons.
	For the same reasons as before, moving modules along such a path will not break connectivity.
	Finally, we address the case when $S_{c'}^*\not\subseteq N^*[c']$.
	By the choice of $c'$ (that it is the deepest) the children of $c$ satisfy the conditions of the lemma, and $N^*[c']$ is not full.
	Then $c$ must have degree 2 and there is a single empty position in $N^*[c']$, i.e., $N^*[c']\setminus\{N^*[p]\cup N^*[q]\}$ where $p$ and $q$ are the children of $c$, \cref{fig:base-case-small-e}.
	Using the same arguments as above, we can find a path from a module outside $N^*[c']$ to this empty position as shown in the figure.
	That makes a solid $3\times 3$ square of modules centered at $c$.
	We can then make an exoskeleton with $c$ as its root.
\end{proof}

\begin{figure}[htb]
	\centering
	\captionsetup[subfigure]{justification=centering}
	\begin{subfigure}[t]{.12\textwidth}
		\centering
		\includegraphics[page=1]{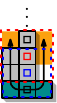}%
		\subcaption{}
		\label{fig:base-case-small-a}
	\end{subfigure}%
	\begin{subfigure}[t]{.22\textwidth}
		\centering
		\includegraphics[page=2]{ipe-basecase-small}%
		\subcaption{}
		\label{fig:base-case-small-b}
	\end{subfigure}%
	\begin{subfigure}[t]{.2\textwidth}
		\centering
		\includegraphics[page=3]{ipe-basecase-small}%
		\subcaption{}
		\label{fig:base-case-small-c}
	\end{subfigure}%
	\begin{subfigure}[t]{.22\textwidth}
		\centering
		\includegraphics[page=4]{ipe-basecase-small}%
		\subcaption{}
		\label{fig:base-case-small-d}
	\end{subfigure}%
	\begin{subfigure}[t]{.24\textwidth}
		\centering
		\includegraphics[page=5]{ipe-basecase-small}%
		\subcaption{}
		\label{fig:base-case-small-e}
	\end{subfigure}%

	\caption{Illustration of~\cref{lem:base-case1}. Modules $d$ and $c$ are highlighted with inner red and blue squares.  The boundaries of $N^*[d']$ and $N^*[c']$ are shown with dashed red and blue lines respectively.}
	\label{fig:base-case-small}
\end{figure}

By applying~\cref{lem:base-case1} a constant number of times in appropriate subskeletons, we eventually obtain a solid $3\times 3$ square from which we can construct an exoskeleton.

\begin{corollary}
    \label{cor:base-case}
    Given a subskeleton $S_c$ with $9\le|S_c^*|$, we can reconfigure $S_c^*$ transforming one of its subskeletons into an exoskeleton in $\mathcal{O}(1)$ many transformations, without changing other exoskeletons or modules not in $S_c^*$.
\end{corollary}

\Cref{cor:base-case} provides the base case of our induction for~\cref{lem:gather}.
We now establish some routines for the inductive step.
We assume there is a module $d\in S_h$ so that for every child~$c$ of $d$, $S_c^*$ was reconfigured to an exoskeleton $X_c$ satisfying Property~\ref{itm:core-4} with~$k=0$.
If $N^*[d]$ is full, we can finish by making $d$ the root of the exoskeleton, effectively merging the children exoskeletons.
Note that this will necessarily happen if $d$ had degree 4 and is not part of a 4-cycle.
Thus, assume that there is an empty position $e\in N^*[d]$.
We~apply a routine \textsc{Inchworm-Push} that fills $e$ with local moves, using a module initially in the core of the exoskeleton $X_c$ (making that cell empty).
The routine \textsc{Inchworm-Pull} then moves core modules higher (in tree metric) pushing empty positions deeper until they ``exit'' the exoskeleton at leaves.
Again, ignoring nonlocal interactions, these operations are straightforward. However, in the general case, these interactions mandate extra care.

\begin{itemize}
    \item \label{def:inchworm-push}
    \textsc{Inchworm-Push}:
    Let $c$ be a child of $d$, and $p$ be the parent of $d$.
    Require that $N^*[c]$ is full.
    We branch into cases. (i) If $c$, $d$, and $p$ are collinear, there is a path $\pi$ in $N^*(d)$ from $c$ to~$e$ going through only occupied cells, see~\cref{fig:inchworm-push-i}; (ii) If $c$, $d$ and $p$ form a bend, there is a path $\pi$ in~$N^*[d]$ from $c$ to $e$ going through $d$ and at most one nonskeleton module, see~\cref{fig:inchworm-push-ii,fig:inchworm-push-ii-up}; (iii)~In analogous cases to (i--ii), if a node in $\{c,d,p\}$ is a 4-cycle, there is a path $\pi$ in $N^*(d)$ from a module in $c$ to $e$ going through one module in the shell of $X_d$, depicted in~\cref{fig:inchworm-push-iii}.
    For all cases,
    move the modules along $\pi$ making $e$ occupied and $c$ empty.
    This requires at most two transformations since there is at most one collision along $\pi$ and we can decompose $\pi$ into two paths with no collisions.

    \begin{figure}[ht]
        \captionsetup[subfigure]{justification=centering}%
        \begin{subfigure}[t]{.12\columnwidth}%
            \centering%
            \includegraphics[page=1]{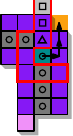}%
            \subcaption{}
            \label{fig:inchworm-push-i}
        \end{subfigure}%
        \hfill%
        \begin{subfigure}[t]{.45\columnwidth}%
            \centering%
            \includegraphics[page=2]{ipe-inchworm-push-i}%
            \subcaption{}
            \label{fig:inchworm-push-ii}
        \end{subfigure}%
        \hfill%
        \begin{subfigure}[t]{.15\columnwidth}%
            \centering%
            \includegraphics[page=3]{ipe-inchworm-push-i}%
            \subcaption{}
            \label{fig:inchworm-push-ii-up}
        \end{subfigure}%
        \hfill%
        \begin{subfigure}[t]{.15\columnwidth}%
            \centering%
            \includegraphics[page=4]{ipe-inchworm-push-i}%
            \subcaption{}
            \label{fig:inchworm-push-iii}
        \end{subfigure}%
        \caption{\textsc{Inchworm-Push}. Module $c$ is colored turquoise, cell $e$ is yellow, and $d$ is highlighted with an inner triangle. The red curve encloses a subconfiguration around the moves that remains connected guaranteeing that the transformations do not disconnect the configuration.}
        \label{fig:inchworm-push}
    \end{figure}

    \item \label{def:inchworm-pull}
    \textsc{Inchworm-Pull}:
    The operation consists of four stages, each involving at most two transformations.
    For every empty cell $p$ in the core, select one of its children $q$ arbitrarily, and let $x(q)$ be its $x$-coordinate.
    Move the module from $q$ to $p$ in stage $j\in \{1,2,3,4\}$ if ${x(q)\equiv j \pmod 4}$.
    \Cref{fig:inchworm-pull-a} gives two examples, in one of which $p$ is a 4-cycle (that takes two transformations).
    If $q$ is a leaf, if there is a tail module of $q$ that is originally in~$S_c^*$ and not in another exoskeleton, move such a module to $q$, visualized in~\cref{fig:inchworm-pull-b,fig:inchworm-pull-c}.
    If there are no more tail modules we can update~$\overline{X_c}$ by deleting $q$ which causes the tail module that just moved to become a shell module, see~\cref{fig:inchworm-pull-b}, or causes two shell modules to become tail modules as in~\cref{fig:inchworm-pull-c}.

    \begin{figure}[ht]
    	\centering%
    	\captionsetup[subfigure]{justification=centering}%
    	\begin{subfigure}[t]{.5\textwidth}%
    		\centering%
    		\includegraphics[page=1]{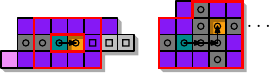}%
    		\subcaption{}
    		\label{fig:inchworm-pull-a}
    	\end{subfigure}%
    	\hfil%
    	\begin{subfigure}[t]{.5\textwidth}%
    		\centering%
    		\includegraphics[page=2]{ipe-inchworm-pull}%
    		\subcaption{}
    		\label{fig:inchworm-pull-b}
    	\end{subfigure}%
    	\par%
    	\begin{subfigure}[t]{.5\textwidth}%
    		\centering%
    		\includegraphics[page=3]{ipe-inchworm-pull}%
    		\subcaption{}
    		\label{fig:inchworm-pull-c}
    	\end{subfigure}%
    	\hfil%
    	\begin{subfigure}[t]{.5\textwidth}%
    		\centering%
    		\includegraphics[page=5]{ipe-inchworm-proof}%
    		\subcaption{}
    		\label{fig:inchworm-pull-d}
    	\end{subfigure}%
    	\caption{\textsc{Inchworm-Pull}. Module $q$ is colored turquoise and empty cell $p$ is yellow.}
    	\label{fig:inchworm-pull}
    \end{figure}
\end{itemize}

The reason we stagger the moves into four stages in \textsc{Inchworm-Pull} is to guarantee that the configuration contains enough static modules around a move.
Otherwise, performing all moves in a single transformation might disconnect the configuration, as shown in~\cref{fig:inchworm-pull-d}.
We now argue that \textsc{Inchworm-Push} and \textsc{Inchworm-Pull} preserve connectivity and require only~$\mathcal{O}(1)$ transformations.
To establish local connectivity of the static modules in the neighborhood of each move, we rely on Properties~\ref{itm:core-3} and~\ref{itm:core-4}.

\begin{lemma}
	\label{lem:inchworm-push}
	\hyperref[def:inchworm-push]{\textsc{Inchworm-Push}} maintains connectivity and takes $\mathcal{O}(1)$ transformations, given that the conditions of the operation are met.
\end{lemma}

\begin{figure}[htb]
	\centering
	\captionsetup[subfigure]{justification=centering}
	\begin{subfigure}[t]{.2\textwidth}
		\centering
		\includegraphics[page=1]{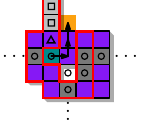}%
		\subcaption{}
		\label{fig:push-proof-a}
	\end{subfigure}%
	\begin{subfigure}[t]{.2\textwidth}
		\centering
		\includegraphics[page=2]{ipe-inchworm-proof}%
		\subcaption{}
		\label{fig:push-proof-b}
	\end{subfigure}%
	\begin{subfigure}[t]{.2\textwidth}
		\centering
		\includegraphics[page=3]{ipe-inchworm-proof}%
		\subcaption{}
		\label{fig:push-proof-c}
	\end{subfigure}%
	\begin{subfigure}[t]{.2\textwidth}
		\centering
		\includegraphics[page=4]{ipe-inchworm-proof}%
		\subcaption{}
		\label{fig:push-proof-d}
	\end{subfigure}%

	\caption{\textsc{Inchworm-Push} operations preserve connectivity.}
	\label{fig:inchworm-proof}
\end{figure}

\begin{proof}
	In all cases of \textsc{Inchworm-Push}, local connectivity around $d$ is never broken since~$N^*[c]$ is full by requirement ($N^*[d]$ remains connected after removing the moving modules).
	Thus, connectivity could only be broken for nonlocal parts of the configuration.
	This happens if the modules in $\pi$ are also part of nonlocal shells.

	Assume the orientation of~\cref{fig:inchworm-push}, i.e.,~$p$ is above $d$ and, in case (i) $c$ is below $d$, and in cases (ii--iii) $c$ is to the left of~$d$.
	Note that if the shell of any nonlocal cell contains a moving module in cases~\cref{fig:inchworm-push-ii,fig:inchworm-push-iii}, then the cell $e$ depicted in the figure must not be empty, since it would also be in this cell's shell.
	Thus, if we assume that a moving module is part of a nonlocal structure, we must be in case~(i), or in case (ii) when $e$ is the cell to the northeast of $c$ (\cref{fig:inchworm-push-i,fig:inchworm-push-ii-up}).

	We argue first for case (i).
	The nonlocal part of the exoskeleton must have a core cell in the southeast of $c$ or two cells to the east of $c$ (or both).
	Let $s$ and $q$ be such cells respectively.
	If the cell to the north of $q$ is also in the core of some exoskeleton, then $e$ is full since it's a shell position, a contradiction.
	If both $s$ and $q$ are leaves then the parent of $s$ is to its south and the parent of $q$ is to its east.
	Then, the moving modules of $\pi$ are not part of a shell and connectivity is not broken.
	So, assume the cell $t$ to the south of $q$ (and east of~$s$) is a core cell.
	If all three of $s$, $q$ and $t$ are full, then connectivity is not broken.
	Note that only one of such cells can be empty by Property~\ref{itm:core-4}.
	If either $s$ or $q$ are empty, then, by Properties~\ref{itm:core-3} and~\ref{itm:core-4}, all cells in $N^*(\{s,q\})$ are full.
	Then connectivity is not broken as shown in~\cref{fig:push-proof-a,fig:push-proof-b}.
	If $t$ is the only empty cell in $N^*(\{s,q,t\})$, then the same argument holds as shown in~\cref{fig:push-proof-c}.
	The only other cell that could be empty is the cell~$u$ to the southeast of $t$, if it is part of the core of an exoskeleton that is not local to $t$.
    By Properties~\ref{itm:core-3} and~\ref{itm:core-4}, $N(\{t,u\})$ is full, and connectivity is not broken (\cref{fig:push-proof-d}).

	Finally, all the arguments above can be repeated for case (ii) if we replace $c$ with $d$.
\end{proof}

\begin{lemma}
	\label{lem:inchworm-pull}
	\hyperref[def:inchworm-pull]{\textsc{Inchworm-Pull}} preserves connectivity and takes $\mathcal{O}(1)$ transformations.
\end{lemma}

\begin{figure}[h]
	\centering
	\captionsetup[subfigure]{justification=centering}
	\begin{subfigure}[t]{.2\textwidth}
		\centering
		\includegraphics[page=6]{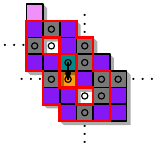}%
		\subcaption{}
		\label{fig:pull-proof-a}
	\end{subfigure}%
	\begin{subfigure}[t]{.2\textwidth}
		\centering
		\includegraphics[page=7]{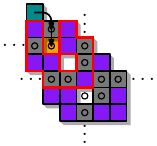}%
		\subcaption{}
		\label{fig:pull-proof-b}
	\end{subfigure}%
	\caption{\textsc{Inchworm-Push} operations preserve connectivity.}
	\label{fig:inchworm-pull-proof}
\end{figure}

\begin{proof}
	It is clear by its definition that \textsc{Inchworm-Pull} executes $\mathcal{O}(1)$ parallel moves.
	Let $p$ and $q$ be one of the empty positions and its chosen child, respectively.
	Similar to \textsc{Inchworm-Push}, when the neighborhood of $\{p,q\}$ does not intersect with nonlocal parts of the configuration, it is straightforward to check that the move from $q$ to $p$ does not break connectivity.
	That is $N^*(\{p,q\})$ is connected by Properties~\ref{itm:core-3} and~\ref{itm:core-4} unless it contains at least two nonlocal empty core cells.
	In such cases, either $p$ is in the shell of a nonlocal part of an exoskeleton and $q$ is a leaf (\cref{fig:pull-proof-b}), or both $p$ and $q$ are in the shell of a nonlocal part of an exoskeleton (\cref{fig:pull-proof-a}).
	Note that a core cell can only be part of a shell if it is a degree-2 bend vertex-adjacent to another degree-2 bend of an exoskeleton core (see $p$ and $q$ in \cref{fig:pull-proof-a} shown in yellow and teal, respectively, that are adjacent to empty positions in the core).
	By the definition of the operation, none of the adjacent empty positions can be part of a move in the same stage.
	Let $M$ be the set of empty core cells vertex-adjacent to $p$ and $q$.
	Then, by Properties~\ref{itm:core-3} and~\ref{itm:core-4},
	$N^*(\{p_i,q_i\}\cup M)$ is connected, and the connectivity is not~broken.
\end{proof}

\begin{lemma}
    \label{lem:gather}
    Given a configuration $C$ with skeleton $S$, and a skeleton node $d$ with $|S_d^*|\ge 9$, $\mathcal{O}(|S_d^*|)$ transformations are sufficient to reconfigure $S_d$ into an exoskeleton $X_d$.
\end{lemma}
\begin{proof}
    If a subskeleton of $S_d$ does not contain an exoskeleton, we can apply~\cref{cor:base-case}.
    If,~after that, $S_d$ is not the core of an exoskeleton $X_d$, we apply induction.
    Let $Q$ be the set containing the children of $d$.
    For every child $c\in Q$ with $|S_c^*|\ge 9$ we can get an exoskeleton~$X_c$ in $\mathcal{O}(|S_c^*|)$ transformations by inductive hypothesis.
    For every child $c\in Q$ with $|S_c^*|< 9$, we apply \cref{lem:base-case1} that either results in an exoskeleton $X_c$, or results in deleting $c$ from~$S$ ``compacting'' all its modules in $N^*[c]\cap N^*[d]$.
    If after that $N^*[d]$ is not full, $N^*[d]$ contains at most three empty cells if $d$ is a degree-2 bend in $S$, at most two empty cells if $d$ is a ``straight'' degree-2 in $S$, or at most 1 empty cell if $d$ has degree 3.
    (Note that if $d$ has degree~4 and is not a 4-cycle, then $N^*[d]$ is full.)
    For each of these empty cells, we can fill them by applying \textsc{Inchworm-Push} followed by at most four applications of \textsc{Inchworm-Pull} to ensure that~$N^*[c]$ is full for all children $c$ of $d$ while maintaining Property~\ref{itm:core-4}.
    This takes at most~$\mathcal{O}(1)$ transformations.
    When $N^*[d]$ is full, the configuration contains an exoskeleton~$X_d$ except for the ``4-arity'' of the empty positions (Property~\ref{itm:core-4}).
    That can be resolved by applying \textsc{Inchworm-Pull} to $X_c$ at most three times for each child $c$ of~$d$.
\end{proof}

By Lemma~\ref{lem:heavy-node}, we can choose an appropriately heavy node $d$ of $S$ to obtain:

\begin{corollary}
    Let $C$ be a configuration with bounding box perimeter $P$. We can reconfigure~$C$ so that it contains an exoskeleton $X_h$ with $\ge 36P$ modules in $\mathcal{O}(P)$ transformations.
    All~modules stay within $\mathcal{O}(1)$ distance from the bounding box of $C$.
\end{corollary}

\subsection{Phase~(II): Scaffolding}
\label{subsec:scaffolding}

\phaseref{Phase~(I)} of our algorithm recursively constructed a sufficiently large exoskeleton $X_h$ that is now reconfigured in \phaseref{Phase~(II)}.
The goal of this reconfiguration is to create a \mbox{$\boldsymbol{\dashv}$ -shaped} exoskeleton, with its vertical segment precisely aligned along the right boundary of the bounding box.
This vertical edge will serve as the sweep line in \phaseref{Phase~(III)}.

\phaseref{Phase~(II)} of our approach consists of the following three steps:
\medskip
\begin{enumerate}
    \item \label{shoot-1} First ``compact'' $X_h$ so that its core contains no empty cells.
    \item \label{shoot-2} We then choose a rightmost node $c$ in the core of $X_h$ and use \textsc{Inchworm-Push} and \textsc{-Pull} to grow a horizontal path from $c$ rightward, until a $3\times 3$ square is formed outside the bounding box, see~\crefrange{fig:shoot-a}{fig:shoot-f}.
    \item \label{shoot-3} Use \textsc{Inchworm-Push} and \textsc{-Pull} to grow the path upwards until it hits the top of the bounding box, then grow a path downwards until it hits the bottom of the bounding box, see~\crefrange{fig:shoot-f}{fig:shoot-h}.
\end{enumerate}

Step~\ref{shoot-1} can be accomplished by applying \textsc{Inchworm-Pull} in $X_h$ until all empty core cells are gone, which is up to the height of $X_h$ many times, hence $\mathcal{O}(P)$ parallel moves.

We now describe Step~\ref{shoot-2}.
The main technical difficulty is that the path will intersect nonlocal parts of the configuration. These will not be part of $X_h$ by the choice of $c$.
We resolve this by using the connectivity properties of exoskeletons to ``go through'' the obstacles without breaking connectivity.
Let $d$ be $c$'s east neighbor.
From now on, we consider the exoskeleton rooted at $c$, $X_c$, and note that $h$ becomes one of the leaves of $\overline{X_c}$.
We ensure that throughout the reconfiguration, $\overline{X_c}$ always contains $h$ as a leaf so that connectivity is maintained with the parts of the original skeleton $S$ that are not part of the exoskeleton $X_c$.
To accomplish that, whenever we apply \textsc{Inchworm-Pull} to $X_c$, we choose, if possible, a child $q$ of an empty position $p$ that is not an ancestor of $h$.

Our goal now is to fill the three positions to the right (NE, E, and SE) of $d$.
If all are full, we move the labels $c$ and $d$ to the right.
Else, if $d$ is not a cut vertex, we try moving $d$ to the right using \textsc{Inchworm-Push} case (ii) (\cref{fig:shoot-a}).
Else, the position to the E~of~$d$ is full and we verify whether applying  \textsc{Inchworm-Push} case (i) is feasible.
If so, we perform it.

The remaining case is when, without loss of generality, we wish to fill the NE cell from $d$ and the module $p$, which is either N or NW of $d$, is a cut vertex (\cref{fig:shoot-b}).
Then, we first fill the cells $(0,2)+c$ and $(-1,2)+c$ (shown in red in~\cref{fig:shoot-b}).
We first check $(0,2)+c$, if it is empty, we move the two modules below it up.
That makes $c$ empty. We can then apply \textsc{Inchworm-Pull} four times to $X_c$, making it an exoskeleton again.
Repeat the same for $(-1,2)+c$.
Now, $p$ is no longer a cut vertex and we can apply \textsc{Inchworm-Push} case~(i).
After that, $c$  is again empty.
Instead of using \textsc{Inchworm-Pull}, we can route the two modules who just filled $(0,2)+c$ and $(-1,2)+c$ to $c$ and the empty position at distance~4 away from~$c$ in $\overline{X_c}$ created by the four executions of \textsc{Inchworm-Pull} (\cref{fig:shoot-d,fig:shoot-e}).
This makes $X_c$ an exoskeleton again.

Eventually, $c$ moves 2 positions to the right of the bounding box, and $X_c$ has a $3\times 3$ square of full cells outside the bounding box.
Step 3 is very similar to Step 2.
First relabel the module above $c$ as $d$ and propagate the path upwards in the same way (\cref{fig:shoot-f}), except that $d$ will never be a cut vertex and we can always use \textsc{Inchworm-Push} case (ii).
Once the top edge is reached, relabel $c$ and $d$ accordingly to propagate the path downwards~(\cref{fig:shoot-g}).
Note that after the relabeling, $X_c$ is an exoskeleton except for the ``4-arity'' of the empty core cells.
They can be fixed with $\mathcal{O}(1)$ applications of \textsc{Inchworm-Pull}.

\begin{lemma}
	\label{lem:shooting}
	\textsc{Scaffolding} takes $\mathcal{O}(P)$ transformations and does not break connectivity.
\end{lemma}
\begin{proof}
	The ``$\dashv$-shaped'' structure created in this phase has size $\mathcal{O}(P)$ and, thus, require $\mathcal{O}(P)$ applications of \textsc{Inchworm-Push} and \textsc{Inchworm-Pull}.
	When applying \textsc{Inchworm-Push} case (i), the extra moves guarantee that connectivity is not broken locally.
	At any moment, every component of $C\setminus X_c$ is either connected to the neighborhood of $h$, or to the shell of~$X_c$.
	Whenever a path of $X_c$ shrinks due to an application of \textsc{Inchworm-Pull} we leave the cells in nonlocal parts of the original skeleton, so components of $C\setminus X_c$ potentially merge, but still are either connected to the neighborhood of $h$, or to the shell of $X_c$.
\end{proof}

\begin{figure}[htb]
    \captionsetup[subfigure]{justification=centering}%
    \begin{subfigure}[t]{.333\textwidth}%
        \centering%
        \includegraphics[page=1]{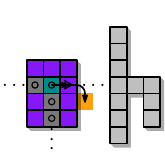}%
        \subcaption{}
        \label{fig:shoot-a}
    \end{subfigure}%
    \begin{subfigure}[t]{.333\textwidth}%
        \centering%
        \includegraphics[page=2]{ipe-shooting}%
        \subcaption{}
        \label{fig:shoot-b}
    \end{subfigure}%
    \begin{subfigure}[t]{.333\textwidth}%
        \centering%
        \includegraphics[page=3]{ipe-shooting}%
        \subcaption{}
        \label{fig:shoot-c}
    \end{subfigure}%
    \par\bigskip\hfil%
    \begin{subfigure}[t]{.333\textwidth}
        \centering%
        \includegraphics[page=4]{ipe-shooting}%
        \subcaption{}
        \label{fig:shoot-d}
    \end{subfigure}%
    \begin{subfigure}[t]{.333\textwidth}
        \centering%
        \includegraphics[page=5]{ipe-shooting}%
        \subcaption{}
        \label{fig:shoot-e}
    \end{subfigure}%
    \par\bigskip%
    \begin{subfigure}[t]{.33\textwidth}
        \centering%
        \includegraphics[page=6]{ipe-shooting}%
        \subcaption{}
        \label{fig:shoot-f}
    \end{subfigure}%
    \begin{subfigure}[t]{.33\textwidth}
        \centering%
        \includegraphics[page=7]{ipe-shooting}%
        \subcaption{}
        \label{fig:shoot-g}
    \end{subfigure}%
    \begin{subfigure}[t]{0.33\textwidth}
        \centering%
        \includegraphics[page=8]{ipe-shooting}%
        \subcaption{}
        \label{fig:shoot-h}
    \end{subfigure}%
    \caption{Building the scaffolding; bounding box of the initial configuration is shown in~yellow.}
    \label{fig:shooting}
\end{figure}
    \subsection{\phaseref{Phase~(III)}: Creating scale}
\label{subsec:sweep}
Once \phaseref{Phase~(II)} concludes, we are left with an intermediate configuration that has an exoskeleton spanning the east edge of its bounding box, recall~\cref{fig:shoot-h}.
In the final phase of our algorithm, we use this to form a $3$-scaled configuration that can be arbitrarily reconfigured due to~\cref{thm:3-scaled-reconfiguration}.
We start by modifying part of the exoskeleton structure created during \phaseref{Phase~(II)} into a vertical line of $3\times 3$ \newterm{meta-modules}.

\subparagraph*{Meta-modules.} For the remainder of this section, a meta-module $M$ consists of eight modules occupying the open vertex-neighborhood of a common center cell $v$, i.e., one in every cell of $N^*(v)$ as shown in~\cref{fig:meta-module}.
We define as its east and west strips as the cells which share a $y$-coordinate with one of its modules and are located within the bounding box of the full configuration $C$.
We denote these by $\east(M)$ and $\west(M)$, respectively.

A \newterm{sweep line} is then a connected subconfiguration $\ell=C[x(\ell)-1,x(\ell)+1]^x$ formed by disjoint meta-modules that spans the full height of the bounding box of $C\setminus\ell$.
We label these meta-modules as $M_1,\ldots, M_h$, where $v_i$ is the center cell of $M_i$ such that $y(v_i)<y(v_j)\Leftrightarrow i<j$, see~\cref{fig:sweep-line,fig:sweep-line-separator}.
Note that $x(v_i)=x(\ell)$ for every center cell, and the east and west strips of $\ell$ correspond exactly to the union of meta-module strips.
We use these to define the following characteristics.

\begin{figure}[htb]%
    \hfil%
    \begin{subfigure}[t]{0.23\textwidth}
        \centering%
        \includegraphics[page=1]{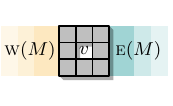}%
        \subcaption{A meta-module $M$.}%
        \label{fig:meta-module}%
    \end{subfigure}%
    \hfil%
    \begin{subfigure}[t]{0.28\textwidth}%
        \centering%
        \includegraphics[page=2]{sweep-line-and-meta-modules}%
        \subcaption{A clean sweep line (red).}%
        \label{fig:sweep-line}%
    \end{subfigure}%
    \hfil%
    \begin{subfigure}[t]{0.3\textwidth}%
        \centering%
        \includegraphics[page=3]{sweep-line-and-meta-modules}%
        \subcaption{A separator sweep line (red).}%
        \label{fig:sweep-line-separator}%
    \end{subfigure}%
    \caption{Sweep lines are formed by vertical stacks of meta-modules.}
\end{figure}

A sweep line $\ell$ is \newterm{clean} exactly if every meta-module $M_i$ either has an empty center cell~$v_i$, or its respective west strip is fully occupied, i.e, $\west(M_i)\subseteq C$.
See~\cref{fig:sweep-line} for an illustration.
Analogously, we call~${\ell}$ \newterm{solid} if all center cells of its meta-modules are occupied.

Furthermore, we say that is ${\ell}$ a \newterm{separator} if for every meta-module $M_i$ in ${\ell}$:
\begin{enumerate}[1.]
    \item The east strip $\east(M_i)$ does not contain modules unless all cells in $\west(M_i)$ are full.
    \item The modules in $\east(M_i)$ are located in its $x$-minimal cells, thus forming a $3$ cell tall rectangle with at most two additional ``loose'' modules.
\end{enumerate}

We start by transforming part of the exoskeleton from \phaseref{Phase~(II)} into meta-modules that form a sweep line at the eastern boundary of $B$, as illustrated in~\cref{fig:skeleton-to-sweep-line}.
This can be achieved by moving only modules that occupy the core of the relevant section of the skeleton.
\begin{figure}[htb]
    \captionsetup[subfigure]{justification=centering}%
    \centering%
    \begin{subfigure}[t]{99.2pt}
        \transformable{\includegraphics[page=1]{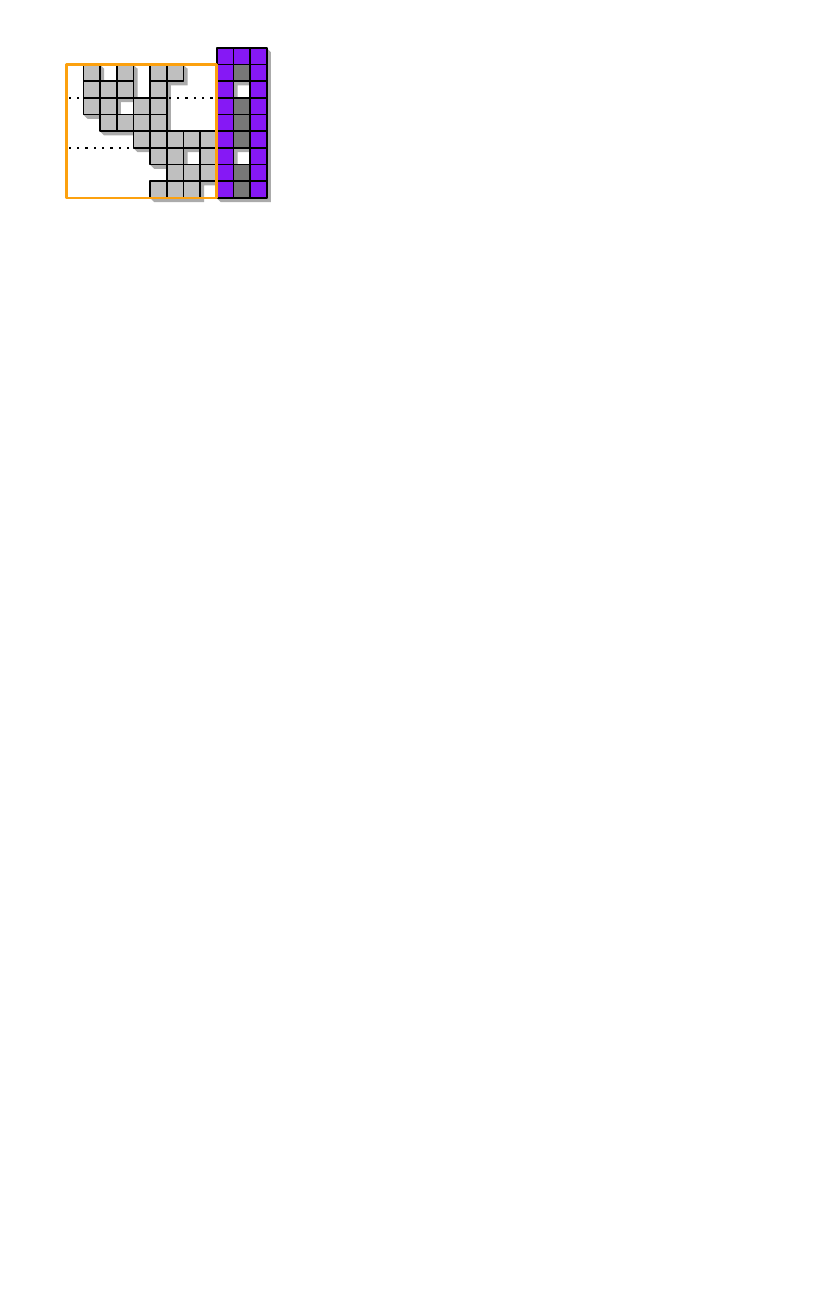}}%
        \subcaption{}%
    \end{subfigure}%
    \ttransforms%
    \begin{subfigure}[t]{99.2pt}
        \includegraphics[page=2]{sweep-overview}%
        \subcaption{}%
    \end{subfigure}%
    \caption{We identify a vertical portion of the skeleton (a) that can be transformed into a sweep line (b) by shifting modules in its interior.}%
    \label{fig:skeleton-to-sweep-line}
\end{figure}
\begin{lemma}
    \label{lem:exoskeleton-to-sweep-line}
    The vertical portion of the ``$\dashv$-shaped'' exoskeleton in $C$ can be made into a separator sweep line in $\mathcal{O}(1)$ transformations.
\end{lemma}
\begin{proof}
    Consider the scaffold exoskeleton $X_c$ as defined in~\cref{subsec:scaffolding}, and let $v_1\in \overline{X_c}$ refer to its $xy$-minimal core cell.
    Due to \cref{lem:shooting}, $C$ has a subskeleton $X_{\ell}$ that contains the modules between the west edge of the bounding box $B$ and the west edge of the extended bounding box $B'$, such that the shell of $X_{\ell}$ is fully occupied.
    This means that the modules in $C\cap\overline{X_{\ell}}$ are free.

    Let $B_y$ refer to the height of the minimal bounding box of $C$.
    Due to the $4$-arity of the depth of unoccupied cells in skeletons, at least $\nicefrac{3}{4}B_y-2$ cells in $\overline{X_{\ell}}$ contain a module.
    To obtain a sweep line, every core cell $v\in\overline{X_{\ell}}$ that has $y(v) \not\equiv 1 (\bmod3)$ must be occupied.
    This requires exactly $\nicefrac{2}{3}B_y-2$ modules, so a sequence of at most $B_y$ chain moves in $\overline{X_{\ell}}$ suffices to reconfigure $X_{\ell}$ into a sweep line of $C$.
    Clearly, this sweep line is a separator; due to the initial construction of $X_c$, its west half-space is empty.
\end{proof}

This gives us a separator sweep line ${\ell}$ that is not necessarily clean or solid, located at the east edge of the configuration's bounding box $B$.
Our goal is to move $\ell$ towards the west edge of $B$ while maintaining its separator property, compacting all encountered modules into $3$-tall strips in the process.
We define two procedures that suffice to reach this goal.

Intuitively, the \textsc{Advance} operation provides a schedule that translates a sweep line by one unit to the west, and the \textsc{Clean} operation
replaces a given sweep line with a clean one.
To efficiently parallelize these operations, we split the sweep line into \newterm{leading} and \newterm{trailing} meta-modules such that $M_i$ is trailing exactly if $i$ is odd, and leading otherwise.
At any given time, either the leading or the trailing meta-modules are active, but never both.

By repeatedly alternating between the \textsc{Advance} and \textsc{Clean} operations, we create a configuration with a clean separator sweep line that has an empty west half-space, see~\cref{fig:sweep-overview}.

\begin{figure}[htb]%
    \hfil%
    \includegraphics[scale=1.25,page=3]{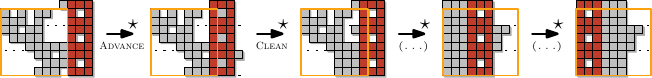}%
    \caption{We alternate between the \textsc{Advance} and \textsc{Clean} operations until $\west(\ell)\cap C=\varnothing$.}%
    \label{fig:sweep-overview}%
\end{figure}

We first argue the \textsc{Clean} operation can be performed in constantly many transformations.
Let $M_i$ be an active meta-module with an occupied center cell.
On a high-level, we consider a rectangle induced by occupied cells in west direction within its strip.
It is easy to see that we can fill an empty position immediately adjacent to this rectangle, while also cleaning the center cell of $M_i$ by at most two subsequent transformations.
By doing this for leading and trailing meta-modules separately, we conclude the following lemma.
\begin{lemma}
    \label{lem:cleaning-sweep-line}
    A configuration $C$ with a sweep line ${\ell}$ can be cleaned by $4$ strictly in-place transformations, without moving any modules east, or exceeding the bounding box $B$ of $C$.
\end{lemma}
\begin{proof}
    Consider a sweep line ${\ell}$ in a configuration $C$ and let $M_i$ refer to an arbitrary meta-module in ${\ell}$ with an occupied center cell~$v_i$.
    We show that $M_i$ can be cleaned in at most $2$ transformations, if neither of its neighbors in ${\ell}$ is being cleaned simultaneously.
    As a result, the following procedure can be used to clean all leading (resp., trailing) modules in parallel, resulting in a schedule of total makespan at most $4$.

    Consider the $xy$-minimal empty cell $v_\varnothing\in\west(M_i)$ and denote the distance that separates it from~$M_i$ by ${k=x(v_i)-x(v_\varnothing)-1}$.
    All cells in $\west(M_i)$ with $x$-coordinates larger than $x(v_\varnothing)$ are thus full, inducing a rectangle~$R$ of $k\times 3$ modules.
    Note that we can assume without loss of generality that such a cell $v_\varnothing$ exists, as $M_i$ does not affect whether ${\ell}$ is clean otherwise.

    We differentiate based on the three $x$-maximal cells in $\west(M_i)\setminus R$, denoted by $\west^k_1,\west^k_2,\west^k_3$ in south to north order, see~\cref{fig:clean-sweep-line-variant-a}.
    Assume for now that $i\in[2,h-1]$, i.e., that $M_i$ has two neighbors in the sweep line ${\ell}$.

    If $\west^k_2\notin C$, the modules in $R$ that have the same $y$-coordinate as $v_i$ are free.
    We thus perform one of the two chain moves shown in~\cref{fig:clean-sweep-line-variant-a} from $v_i$ to either $\west^k_2$ or $\west^k_3$.
    \begin{figure}[htb]
        \begin{minipage}[t][78.5pt][c]{90pt}%
            \includegraphics[page=1]{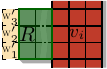}%
        \end{minipage}%
        \begin{minipage}[t]{\columnwidth-90pt}
            \begin{subfigure}[t]{\textwidth}%
                \transformable{\includegraphics[page=2]{sweep-line-clean}}%
                \transforms\includegraphics[page=3]{sweep-line-clean}%
                \hspace{24.5pt}%
                \transformable{\includegraphics[page=4]{sweep-line-clean}}%
                \transforms\includegraphics[page=5]{sweep-line-clean}%
                \subcaption{Trivial cases: $\west^k_2\notin C$.}%
                \label{fig:clean-sweep-line-variant-a}
            \end{subfigure}%
            \par\vspace{0.5em}%
            \begin{subfigure}[t]{\textwidth}%
                \transformable{\includegraphics[page=6]{sweep-line-clean}}%
                \transforms{\includegraphics[page=7]{sweep-line-clean}}%
                \transforms\includegraphics[page=8]{sweep-line-clean}%
                \subcaption{Not all adjacent cells north of $R$ are occupied, and $\west^k_2\in C$.}
                \label{fig:clean-sweep-line-variant-b}
            \end{subfigure}%
            \par\vspace{0.5em}%
            \begin{subfigure}[t]{\textwidth}%
                \transformable{\includegraphics[page=9]{sweep-line-clean}}%
                \transforms{\includegraphics[page=10]{sweep-line-clean}}%
                \transforms\includegraphics[page=11]{sweep-line-clean}%
                \subcaption{All adjacent cells north of $R$ are occupied, and ${\west^k_2\in C}$.}
                \label{fig:clean-sweep-line-variant-c}
            \end{subfigure}
        \end{minipage}
        \caption{``Cleaning'' schedules that move only west or north. These are independent of $\west_3$.}
        \label{fig:clean-sweep-line}
    \end{figure}

    If $\west^k_2\in C$, the same set of modules may not be free.
    In particular, this is true of the module immediately east of $\west^k_2$.
    By definition, at least one of $\west^k_1$ and $\west^k_3$ is unoccupied.
    Assume without loss of generality that $\west^k_3\notin C$.
    The following procedure can be mirrored vertically for the case that $\west^k_3\in C$, as then $\west^k_1\notin C$.

    If there is at least one unoccupied cell adjacent immediately north of $R$, it follows that its neighbor module in $R$ is free.
    Let $v_\north$ refer to the $x$-maximal empty cell that meets these criteria.
    We perform two chain moves, placing a module in $v_\north$, see~\cref{fig:clean-sweep-line-variant-b}.
    If there is no such empty cell $v_\north$, we perform the two moves in \cref{fig:clean-sweep-line-variant-c}, occupying $\west^k_3$.

    It remains to argue that $M_1$ and $M_h$ can be cleaned in a similar manner.
    We argue for the northernmost meta-module $M_h$, but the same line of argumentation applies to $M_1$.
    If~$\west^k_2\notin C$, the schedules depicted in~\cref{fig:clean-sweep-line-variant-a} are immediately applicable.
    Otherwise, we do not fill positions north of $R$, but only south, otherwise proceeding in the same manner.
    In particular, the northernmost modules in $R$ are always free if $\west^k_3\notin C$.
\end{proof}

It remains to argue that the same is true of the \textsc{Advance} operation.
For this, we consider the three positions immediately west of an active meta-module, and distinguish eight cases based on occupied and empty cells.
For all of them, we define local transformation schedules that enable us to move a meta-module one step to the east.
Thus, we obtain:
\begin{lemma}
    \label{lem:advance-sweep}
    A configuration $C$ with a clean separator sweep line ${\ell}$ can be advanced into a clean separator sweep line by~$\mathcal{O}(1)$ strictly in-place transformations if $C\cap\west({\ell})\neq\{\varnothing\}$.
\end{lemma}

\begin{proof}
    We define schedules that allow non-adjacent meta-modules to move westward in parallel and use these to shift all meta-modules west in two phases, as illustrated in~\cref{fig:sweep-overview}.

    Let $M_i$ refer to an arbitrary meta-module in ${\ell}$ with an occupied center cell~$v_i$.
    We show that~$M_i$ can be translated to the west by one unit by a schedule of makespan at most~$6$, if neither of its neighbors in ${\ell}$ is being moved simultaneously.
    The following procedures can thus be used to shift all leading (resp., trailing) modules in parallel.
    Note that the first claim only handles meta-modules which have two adjacent meta-modules in ${\ell}$.
    For the northernmost (resp., southernmost) meta-module in ${\ell}$, we give a separate case analysis.

    \begin{claim}
        We can translate any meta-module $M_i$ with $i\in[2,h-1]$ to the west by one~cell in at most $6$ transformations, without moving any module east.
        \label{clm:locally-advance-two-neighbors}
    \end{claim}
    \begin{claimproof}
        We differentiate based on modules in the cells with $x$-coordinate exactly $x(v_i)-2$, which we denoted by $\west_1,\west_2,\west_3$ in south to north order.
        There exist eight possible configurations, which we identify by the power set of $\{\west_1,\west_2,\west_3\}$ and group into six equivalence classes according to symmetry along the $x$-axis, see~\cref{fig:sweep-advance-east-cells}.
        As we cannot be certain that modules in $\{\west_1,\west_2,\west_3\}$ are free, we mark them with $\vcenter{\hbox{\includegraphics{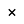}}}$ and avoid moving them.
        \begin{figure}[htb]%
            \includegraphics[page=1]{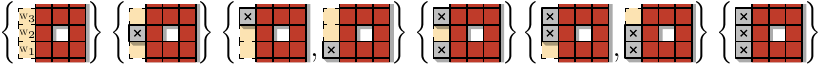}%
            \caption{There are $2^3=8$ possible configurations of the cells $\west_1,\west_2,\west_3$, all shown here.}%
            \label{fig:sweep-advance-east-cells}%
        \end{figure}%

        We start by identifying and handling two trivial cases.
        Firstly, if every cell west of~$M_i$ is already occupied, i.e., $\west(M_i)\subseteq C$, we simply slide the module located immediately west of~$v_i$ to the east by one unit, inducing a clean meta-module at $v_i'=(x(v_i)-1, y(v_i))$.

        Secondly, if $C\cap\{\west_1,\west_2,\west_3\}=\varnothing$, i.e., none of the west-adjacent cells are occupied, the schedule in~\cref{fig:single-meta-module-translate-leading} can easily be verified:
        No collisions occur, and a connected backbone is maintained.
        While the illustrated schedule is applicable only if $M_i$ is leading, we can simply reverse and rotate the moves by~$180^\circ$, in case $M_i$ is trailing.

        \begin{figure}[hp]
            \begin{subfigure}[t]{\textwidth}%
                \hspace{10.75pt}%
                \transformable{\includegraphics[page=1]{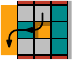}}%
                \transforms{\includegraphics[page=2]{sweep-advance-single-trivial}}%
                \transforms{\includegraphics[page=3]{sweep-advance-single-trivial}}%
                \transforms{\includegraphics[page=4]{sweep-advance-single-trivial}}%
                \transforms{\includegraphics[page=5]{sweep-advance-single-trivial}}%
                \transforms{\includegraphics[page=6]{sweep-advance-single-trivial}}%
                \subcaption{%
                    Leading, $\west_1$, $\west_2$, and $\west_3$ are empty.
                    This schedule can be rotated and reversed for trailing modules.
                }
                \label{fig:single-meta-module-translate-leading}
            \end{subfigure}%
            \par\vspace{0.5em}%
            \hfil\rule{\columnwidth-2em}{0.5pt}%
            \par\vspace{0.5em}%
            \begin{subfigure}[t]{\textwidth}%
                \transformable{\includegraphics[page=1]{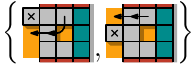}}%
                \transforms{\includegraphics[page=2]{sweep-advance-single-push}}%
                \transforms{\includegraphics[page=3]{sweep-advance-single-push}}%
                \transforms{\includegraphics[page=4]{sweep-advance-single-push}}%
                \transforms{\includegraphics[page=5]{sweep-advance-single-push}}%
                \subcaption{Leading, $|C\cap\{\west_1,\west_2,\west_3\}|=1$.}%
                \label{fig:sweep-advance-single-leading-one}%
            \end{subfigure}%
            \par\vspace{6pt}%
            \begin{subfigure}[t]{227.75pt}%
                \transformable{\includegraphics[page=6]{sweep-advance-single-push}}%
                \transforms{\includegraphics[page=7]{sweep-advance-single-push}}%
                \transforms{\includegraphics[page=8]{sweep-advance-single-push}}%
                \subcaption{Leading, $|C\cap\{\west_1,\west_2,\west_3\}|=2$.}%
                \label{fig:sweep-advance-single-leading-two}%
            \end{subfigure}%
            \begin{subfigure}[t]{\textwidth-230pt}%
                \hspace{8pt}%
                \transformable{\includegraphics[page=9]{sweep-advance-single-push}}%
                \transforms{\includegraphics[page=10]{sweep-advance-single-push}}%
                \subcaption{Leading, $\{\west_1,\west_2,\west_3\}\subset C$.}%
                \label{fig:sweep-advance-single-leading-three}%
            \end{subfigure}%
            \par\vspace{0.5em}%
            \hfil\rule{\textwidth-2em}{0.5pt}%
            \par\vspace{0.5em}%
            \begin{subfigure}[t]{\textwidth}
                \transformable{\includegraphics[page=11]{sweep-advance-single-push}}%
                \transforms{\includegraphics[page=12]{sweep-advance-single-push}}%
                \transforms{\includegraphics[page=13]{sweep-advance-single-push}}%
                \transforms{\includegraphics[page=14]{sweep-advance-single-push}}%
                \transforms{\includegraphics[page=15]{sweep-advance-single-push}}%
                \subcaption{Trailing, $|C\cap\{\west_1,\west_2,\west_3\}|=1$.}%
                \label{fig:sweep-advance-single-trailing-one}%
            \end{subfigure}%
            \par\vspace{6pt}%
            \begin{subfigure}[t]{280pt}
                \transformable{\includegraphics[page=16]{sweep-advance-single-push}}%
                \transforms{\includegraphics[page=17]{sweep-advance-single-push}}%
                \transforms{\includegraphics[page=18]{sweep-advance-single-push}}%
                \transforms{\includegraphics[page=19]{sweep-advance-single-push}}%
                \subcaption{Trailing, $|C\cap\{\west_1,\west_2,\west_3\}|=2$.}%
                \label{fig:sweep-advance-single-trailing-two}%
            \end{subfigure}%
            \begin{subfigure}[t]{\textwidth-280pt}
                \hspace{8pt}%
                \transformable{\includegraphics[page=20]{sweep-advance-single-push}}%
                \transforms{\includegraphics[page=21]{sweep-advance-single-push}}%
                \subcaption{Trailing, $\{\west_1,\west_2,\west_3\}\subset C$.}%
                \label{fig:sweep-advance-single-trailing-three}%
            \end{subfigure}%
            \caption{Shifting meta-modules to the west. Schedules for the case that a meta-module has two neighbors:  (a) visualizes the case in which the meta-module does not ``collide'' with other modules; leading meta-modules (b)--(d) and (e)--(g) trailing meta-modules, in case there are occupied cells directly to the west of the meta-module.}
            \label{fig:sweep-advance}
        \end{figure}

        It remains to handle the case in which $C\cap\{\west_1,\west_2,\west_3\}\neq \varnothing$ and at least one cell west of the meta-module in $\west(M_i)$ is unoccupied.
        We handle each case separately, using the schedules depicted in~\cref{fig:sweep-advance}.
        Based on whether $M_i$ is a leading or trailing meta-module, we refer to~\crefrange{fig:sweep-advance-single-leading-one}{fig:sweep-advance-single-leading-three} and~\crefrange{fig:sweep-advance-single-trailing-one}{fig:sweep-advance-single-trailing-three}, respectively.
        Once again, due to their locality, the validity of these schedules can easily be verified:
        As neither $M_{i+1}$ nor $M_{i-1}$ move, a connected backbone will always be preserved, and no collisions can occur, meaning that we have successfully shifted $M_i$ to the west by one unit.
    \end{claimproof}

    \begin{claim}
        \label{clm:locally-advance-one-neighbor}
        We can translate any meta-module $M_i$ with $i\in\{1,h\}$ to the west by one~cell in at most $6$ transformations, without moving any module east.
    \end{claim}
    \begin{claimproof}
        The arguments are similar to the proof of~\cref{clm:locally-advance-two-neighbors}.
        We refer to~\cref{fig:sweep-push-local-end} for an illustration of the local transformations based on occupied and free cells immediately adjacent to the respective meta-module.
        Furthermore, by mirroring this schedule along the $x$-axis, we further obtain a valid schedule for the southernmost meta-module, again for both the leading and trailing case.
    \end{claimproof}

    \begin{figure}[p]
        \begin{subfigure}[t]{\textwidth}%
            \hspace{10.75pt}%
            \transformable{\includegraphics[page=7]{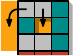}}%
            \transforms{\includegraphics[page=8]{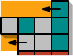}}%
            \transforms{\includegraphics[page=9]{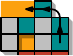}}%
            \transforms{\includegraphics[page=10]{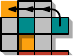}}%
            \transforms{\includegraphics[page=11]{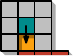}}%
            \transforms{\includegraphics[page=12]{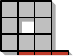}}%
            \subcaption{Leading, $\west_1$, $\west_2$, and $\west_3$ are empty.
            This schedule can be rotated and reversed for trailing modules.}
            \label{fig:single-meta-module-translate-north}
        \end{subfigure}%
        \par\vspace{0.5em}%
        \hfil\rule{\columnwidth-2em}{0.5pt}%
        \par\vspace{0.5em}%
        \begin{subfigure}[t]{\textwidth}%
            \transformable{\includegraphics[page=22]{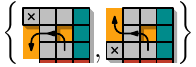}}%
            \hspace{6pt}%
            \transforms{\includegraphics[page=23]{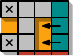}}%
            \transforms{\includegraphics[page=24]{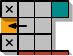}}%
            \transforms{\includegraphics[page=25]{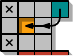}}%
            \transforms{\includegraphics[page=26]{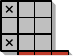}}%
            \subcaption{Leading, $C_\west \in\{\{\west_1\},\{\west_3\}\}$.}%
        \end{subfigure}%
        \par\vspace{0.5em}%
        \begin{subfigure}[t]{\textwidth}%
            \hspace{10.75pt}%
            \transformable{\includegraphics[page=27]{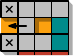}}%
            \transforms{\includegraphics[page=28]{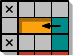}}%
            \transforms{\includegraphics[page=29]{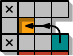}}%
            \transforms{\includegraphics[page=30]{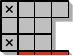}}%
            \subcaption{Leading, $C_\west =\{\west_1,\west_3\}$.}%
        \end{subfigure}%
        \par\vspace{6pt}%
        \begin{subfigure}[t]{242pt}%
            \hspace{10.75pt}%
            \transformable{\includegraphics[page=31]{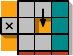}}%
            \transforms{\includegraphics[page=32]{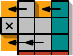}}%
            \transforms{\includegraphics[page=33]{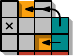}}%
            \transforms{\includegraphics[page=34]{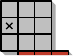}}%
            \subcaption{Leading, $C_\west =\{\west_2\}$.}%
        \end{subfigure}%
        \begin{subfigure}[t]{\columnwidth-242pt}%
            \transformable{\includegraphics[page=42]{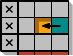}}%
            \transforms{\includegraphics[page=43]{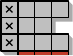}}%
            \subcaption{Leading, $C_\west =\{\west_1,\west_2,\west_3\}$.}%
        \end{subfigure}%
        \par\vspace{6pt}%
        \begin{subfigure}[t]{187pt}%
            \hspace{10.75pt}%
            \transformable{\includegraphics[page=35]{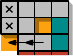}}%
            \transforms{\includegraphics[page=36]{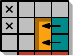}}%
            \transforms{\includegraphics[page=37]{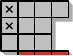}}%
            \subcaption{Leading, $C_\west=\{\west_2,\west_3\}$.}%
        \end{subfigure}%
        \begin{subfigure}[t]{\columnwidth-187pt}%
            \transformable{\includegraphics[page=38]{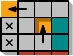}}%
            \transforms{\includegraphics[page=39]{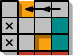}}%
            \transforms{\includegraphics[page=40]{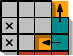}}%
            \transforms{\includegraphics[page=41]{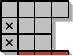}}%
            \subcaption{Leading, $C_\west=\{\west_1,\west_2\}$.}%
        \end{subfigure}%
        \par\vspace{0.5em}%
        \hfil\rule{\columnwidth-2em}{0.5pt}%
        \par\vspace{0.5em}%
        \begin{subfigure}[t]{\textwidth}%
            \transformable{\includegraphics[page=44]{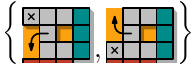}}%
            \transforms{\includegraphics[page=45]{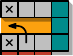}}%
            \transforms{\includegraphics[page=46]{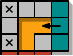}}%
            \transforms{\includegraphics[page=47]{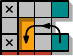}}%
            \transforms{\includegraphics[page=48]{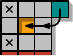}}%
            \transforms{\includegraphics[page=45]{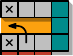}}%
            \subcaption{Trailing, $C_\west \in\{\{\west_1\},\{\west_3\}\}$.}%
        \end{subfigure}%
        \par\vspace{6pt}%
        \begin{subfigure}[t]{242pt}%
            \hspace{10.75pt}%
            \transformable{\includegraphics[page=55]{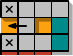}}%
            \transforms{\includegraphics[page=56]{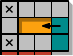}}%
            \transforms{\includegraphics[page=57]{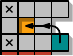}}%
            \transforms{\includegraphics[page=58]{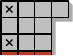}}%
            \subcaption{Trailing, $C_\west=\{\west_1,\west_3\}$.}%
        \end{subfigure}%
        \begin{subfigure}[t]{\textwidth-242pt}%
            \transformable{\includegraphics[page=67]{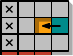}}%
            \transforms{\includegraphics[page=68]{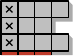}}%
            \subcaption{Trailing, $C_\west=\{\west_1,\west_2,,\west_3\}$.}%
        \end{subfigure}%
        \par\vspace{6pt}%
        \begin{subfigure}[t]{\textwidth}%
            \hspace{10.75pt}%
            \transformable{\includegraphics[page=50]{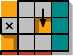}}%
            \transforms{\includegraphics[page=51]{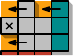}}%
            \transforms{\includegraphics[page=52]{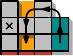}}%
            \transforms{\includegraphics[page=53]{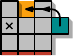}}%
            \transforms{\includegraphics[page=54]{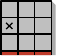}}%
            \subcaption{Trailing, $C_\west=\{\west_2\}$.}%
        \end{subfigure}%
        \par\vspace{6pt}%
        \begin{subfigure}[t]{\textwidth}%
            \hspace{10.75pt}%
            \transformable{\includegraphics[page=63]{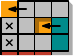}}%
            \transforms{\includegraphics[page=64]{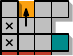}}%
            \transforms{\includegraphics[page=65]{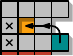}}%
            \transforms{\includegraphics[page=66]{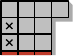}}%
            \subcaption{Trailing, $\C_\west=\{\west_1,\west_2\}$.}%
        \end{subfigure}%
        \par\vspace{6pt}%
        \begin{subfigure}[t]{\textwidth}%
            \hspace{10.75pt}%
            \transformable{\includegraphics[page=59]{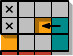}}%
            \transforms{\includegraphics[page=60]{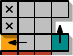}}%
            \transforms{\includegraphics[page=61]{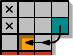}}%
            \transforms{\includegraphics[page=62]{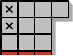}}%
            \subcaption{Trailing, $C_\west=\{\west_2,\west_3\}$.}%
        \end{subfigure}%
        \caption{Shifting meta-modules to the west. Schedules for the case that a meta-module has only one neighbor:
        Let $C_\west=C\cap\{\west_1,\west_2,\west_3\}$.
        Subfigure (a) visualizes the case that $C_\west=\varnothing$.
        Otherwise, leading meta-modules are handled in (b)--(g), and trailing in (h)--(m).}
        \label{fig:sweep-push-local-end}
    \end{figure}

    Due to their strictly local nature, all of our schedules for leading (resp., trailing) meta-modules can be performed in parallel.
    We first shift at all leading modules simultaneously, using the schedules outlined in~\cref{clm:locally-advance-two-neighbors,clm:locally-advance-one-neighbor}.
    This takes at most $6$ transformations.
    Afterward, we apply the same procedures to all trailing meta-modules in parallel, again taking at most $6$ transformations.
    We conclude that it is possible to obtain a sweep line ${\ell}'$ with $x({\ell}')=x({\ell})-1$ in at most $12$ transformations.

    Upon completion, some meta-modules of the sweep line ${\ell}'$ may have a non-empty center cell, and up to two squares per meta-module now additionally occupy the east half-space in a way that violates the separator property.
    We now show that ${\ell}'$ can be made into a clean separator sweep line.

    Consider any meta-module $M'_i$ of ${\ell}'$.
    If $\west(M'_i)\subseteq C$, $M'_i$ is clean by definition, see above.
    Otherwise, there exist at most two modules in $\east(M'_i)$, immediately adjacent to the meta-module itself.
    Due to~\cref{lem:cleaning-sweep-line}, $M'_i$ can be cleaned in $2$ transformations; it only remains to argue that we can restore the separator conditions.
    With a simple chain move, we can place one of the two excess squares in $v'_i$ and clean $M'_i$ once more to free up $v'_i$ again.

    We clean each module at most three times, and all leading (resp., trailing) meta-modules in parallel, taking at most $3\cdot 4+2=14$ additional transformations.
    Afterward, we obtain a clean separator sweep line ${\ell}'$ with $x({\ell}')=x({\ell})-1$.
\end{proof}

After no more than $\mathcal{O}(P)$ many iterations consisting of \textsc{Clean}+\textsc{Advance}, we obtain a state in which $\west(\ell)\cap C = \varnothing$.
The separator property implies that the modules in $\east(\ell)$ form $3$-wide, horizontal stacks.

It remains to show that we can make this configuration $3$-scaled.
On a high level, we achieve this by filling the center cells of all meta-modules and balancing the remaining modules between their east strips such that each contains a multiple of nine.

To simplify our arguments, assume that $n$ is a multiple of nine.
In the general case, $n$ can exceed a multiple of nine by at most $8$ modules.
We can place these at the south-west corner of the extended bounding box $B'$ and locally ensure the existence of a connected backbone during each transformation.
\begin{lemma}
    \label{lem:sweep-line-to-scaled}
    Let $C$ contain a clean separator sweep line ${\ell}$ with an empty west half-space.
    Then $C$ can be transformed into a $3$-scaled configuration in $\mathcal{O}(P)$ transformations.
\end{lemma}

\begin{proof}
    We start by making the sweep line ${\ell}$ in $C$ solid.
    To achieve this, we can simply apply the cleaning operation outlined in~\cref{lem:cleaning-sweep-line} in reverse.
    Even without parallelization, at least one meta-module of ${\ell}$ can be ``un-cleaned'' every $\mathcal{O}(1)$ transformations.
    There are $\mathcal{O}(P)$ many clean meta-modules, so we can obtain a solid sweep line in $\mathcal{O}(P)$ transformations.

    To make $C$ into a $3$-scaled histogram, it remains to balance the modules in $C\setminus{\ell}$ such that for each $M_i\in{\ell}$, the number of modules in $C\cap\west(M_i)$ is a multiple of nine.
    Observe that there are at most $8$ ``loose'' modules in any horizontal strip induced by some meta-module~$M_i$.
    We start with the strip induced by the topmost meta-module $M_h$.
    These modules can be iteratively moved to the strip induced by $M_{h-1}$ by $\mathcal{O}(1)$ transformations as depicted in~\cref{fig:balance-histogram}, and are thereby placed immediately adjacent to occupied cells within this strip.
    Now the number of modules in the strip of $M_h$ is a multiple of $9$, and there are at most~$8$ loose modules in the strip of $M_{h-1}$.
    \begin{figure}[htb]
        \hfil%
        \includegraphics[page=1]{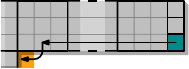}%
        \hfil%
        \includegraphics[page=2]{histogram-balance}%
        \hfil%
        \includegraphics[page=3]{histogram-balance}%
        \caption{The depicted moves can be realized in $\mathcal{O}(1)$ transformations (and in reverse).}
        \label{fig:balance-histogram}
    \end{figure}
    We repeat this procedure until we reach $M_1$, i.e., iteratively move loose modules down, and leave $3$-scaled strips behind.

    Upon completion, there are at most $8$ loose modules located in the bottommost strip.
    These can trivially be moved immediately to the west of the solid sweep line by straightforward chain moves.
    As the bounding box has height $\mathcal{O}(P)$, this approach takes at most $\mathcal{O}(P)$ many~transformations.
\end{proof}

Having constructed a $3$-scaled configuration, we can arbitrarily reconfigure due to~\cref{thm:3-scaled-reconfiguration} before applying \phaseref{Phases~(I-III)} in reverse to obtain our target configuration.

\theoremAlgorithm*

\begin{proof}
    \phaseref{Phases (I-III)} transform any configuration into a $3$-scaled configuration that does not exceed the original bounding box by more than a constant amount, in $\mathcal{O}(P)$ transformations.

    Given an instance $\mathcal{I}=(C_1, C_2)$, we simply apply \phaseref{Phases (I-III)} to both configurations, obtaining two weakly in-place schedules that yield $3$-scaled configurations $C_1'$ and $C_2'$.
    We can compute a schedule $C_1'\rightarrow^* C_2'$ due to~\cref{thm:3-scaled-reconfiguration}, which we concatenate with the reverse schedule of $C_2\rightarrow^* C_2'$.
    The result is a schedule $C_1\rightarrow^*C_2$ of makespan $\mathcal{O}(P)$.
\end{proof}
%
\section{Labeled setting}
\label{sec:labeled-algorithm}

As indicated in the introduction, our reconfiguration methods can also be applied to the labeled setting, albeit with the addition of another fourth phase that handles the labeling.

Recall that the central and only difference between the labeled and the unlabeled model studied in the previous sections is that a configuration $C$ is no longer simply a subset of~$\mathbb{Z}^2$, but an injective mapping $C:[n]\rightarrow\mathbb{Z}^2$.
The image of a such a mapping is then always a valid unlabeled configuration; we call this its \newterm{silhouette}.

The algorithm from~\cref{sec:algorithm} can be used to obtain a specific silhouette with a deterministic (but not necessarily the desired) labeling.
To deal with this, we employ grid-like canonical silhouettes that we call \newterm{sponges} and show that we can efficiently modify their labelings in a strictly in-place manner.

\subparagraph*{Sponges.}
A configuration is a \newterm{$k$-sponge} for $k\geq 3$ exactly if it consists of a rectangle of width and height at least $2k$ in which the cells $(ik,jk)\in\mathbb{Z}^2$ are unoccupied for all $i,j\in \mathbb{N}_0$.
A \newterm{$k$-cell} is then a $k\times k$ square anchored at a multiple of $k$ in both dimensions, see~\cref{fig:sponge}.
The lexicographical ordering of positions in a $k$-cell follows the path shown in the same figure.

\begin{figure}[htb]%
    \begin{minipage}[t]{0.3\columnwidth - 0.5em}%
        \centering%
        \includegraphics[page=1]{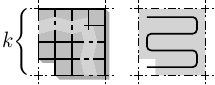}%
        \caption{A single $k$-cell.}%
        \label{fig:sponge}%
    \end{minipage}%
    \hfill%
    \begin{minipage}[t]{0.7\columnwidth - 0.5em}%
        \centering%
        \includegraphics[page=2]{sponge-adjacent-swap}%
        \caption{An adjacent swap using two unoccupied cells.}%
        \label{fig:sponge-adjacent-swap}%
    \end{minipage}%
    \label{fig:combo-figure-sponges}%
\end{figure}

In the following, we argue that for any two labelings of a $k$-sponge silhouette, there exists a schedule with makespan only dependent on the silhouette's bounding box perimeter $P$ and some function in $\poly(k)$ that transforms one labeling into the other.

\medskip

We leverage a classical result from \textsc{Parallel Token Swapping} (or \textsc{Routing Permutations by Matchings}).
This problem was first formulated by Alon, Chung, and Graham~\cite{alon.chung.graham1994routing} and asks, given a graph $G=(V,E)$, for a tight upper bound on the necessary number of \newterm{routing operations} to transform any two vertex labelings of $G$ into one another.
The exact bound is then also called the \newterm{routing number} $\mathsf{rt}(G)$.
Each routing operation selects an arbitrary, independent set of edges $M\subseteq E$ and swaps the labels of vertices that share an edge in $M$.
For the cartesian product $G\times G'$ of two graphs $G$ and $G'$, the following holds.

\begin{theorem}[Alon, Chung, and Graham~{\cite[Theorem 4.3]{alon.chung.graham1994routing}}]
    $\mathsf{rt}(G\times G')\leq 2\mathsf{rt}(G) + \mathsf{rt}(G')$.
    \label{thm:alon-chung-graham}
\end{theorem}
It is an easily validated fact that for the $n$-vertex path graph $P_n$, $\mathsf{rt}(P_n)\leq n$.
Let $G^d$ be a $d$-dimensional grid graph of extent $w_i$ in dimension $i\in [d]$.
Then $G^d = P_{w_1}\times P_{w_2}\ldots \times P_{w_d}$, and its routing number $\mathsf{rt}(G^d)$ is bounded by $\mathcal{O}\bigl(d \sum_{i=1}^d{w_i}\bigr)$.

\medskip

Let $C$ be the silhouette of a sponge with width $w$ and height $h$.
We construct a secondary reconfiguration graph $G^3$, based on which we determine efficient reconfiguration sequences.
Let $G^3$ have as its vertices the modules in $C$ and let its edge set $E$ be defined as follows:
\begin{equation*}
    \forall u,v\in C:\quad \{u,v\}\in E \Leftrightarrow
    \begin{cases}
        \text{$u$ and $v$ are lexicographically adjacent in a $k$-cell of $C$, or}\\
        x(u)=x(v)\pm k\text{, or}\\
        y(u)=y(v)\pm k.
    \end{cases}
\end{equation*}
For a $k$-sponge of width $w$ and height $h$, the secondary reconfiguration graph $G^3$ is now the cartesian product of three path graphs of lengths $\nicefrac{w}{k}$, $\nicefrac{h}{k}$, and $k^2-1$, respectively.
Due to~\cref{thm:alon-chung-graham}, its routing number is then bounded by $\mathcal{O}(\nicefrac{w}{k}+\nicefrac{h}{k}+k^2)$.

\begin{lemma}
    Let $C$ be a $k$-sponge, and $G^3=(C,E)$ its secondary reconfiguration graph.
    Given a matching $M\subset E$, we can swap the matched modules in $\mathcal{O}(k^3)$ transformations.
    \label{lem:matching-swaps-in-sponges}
\end{lemma}

\begin{proof}
    Let $T\subset C$ be any $k$-cell of $C$, and let $M_T\subset M$ correspond to the set of matching edges in $M$ that involve modules of $T$.
    By definition, $\lvert M_T \rvert\in \mathcal{O}(k^2)$.
    We argue that all swaps induced by edges in $M_T$ can be realized by a schedule of makespan in $\mathcal{O}(k^3)$.

    Let $e=\{u,v\}\in M_T$ be an arbitrary edge and fix an arbitrary shortest path embedding of $e$ in the grid.
    This embedding is then fully contained within $T$ and its immediate neighbor $k$-cells.
    To swap two modules, we first translate $u$ along the fixed path towards $v$ by repeated application of the swap schedule shown in~\cref{fig:sponge-adjacent-swap}.
    Once $u$ and $v$ are edge-adjacent, we swap $u$ to its final location and apply the process in reverse, placing $v$ at the original location of the module $u$.
    Each swap operation requires the presence of two of the holes as shown in~\cref{fig:sponge-adjacent-swap}, which can be ensured by shifting the ``holes'' from adjacent $k$-cells to the appropriate location within $T$.
    This takes constantly many transformations.
    To realize just one edge, we thus require at most $\mathcal{O}(k)$ transformations.
    For the entire set $M_T$, the total makespan is therefore in $\mathcal{O}(k^3)$.

    By splitting the $k$-cells of $C$ into nine equivalence classes based on their anchor coordinates modulo $3k$, we can efficiently realize this procedure in parallel, taking at most nine rounds for the entire matching $M$.
\end{proof}

\begin{proposition}
    Let $C_1,C_2$ be two $k$-sponges with identical silhouettes and minimal bounding box perimeter $P$.
    A strictly in-place schedule with makespan $\mathcal{O}(k^2 P + k^5)$ that transforms one into the other can be computed efficiently.
    \label{prop:k-sponge-reconfiguration}
\end{proposition}
\begin{proof}
    Let $C$ be the common silhouette of $C_1$ and $C_2$, and let $G^3$ be the secondary reconfiguration graph of $C$ as outlined previously.
    We can then efficiently compute a sequence of at most $\mathcal{O}(\nicefrac{w}{k}+\nicefrac{h}{k}+k^2)$ routing operations in $G^3$ that transforms the labeling $C_1$ into $C_2$.
    Each routing operation takes at most $\mathcal{O}(k^3)$ transformations to realize in the physical model, so we end up with a total makespan of
     $\mathcal{O}(\nicefrac{w}{k}+\nicefrac{h}{k}+k^2) \cdot \mathcal{O}(k^3) = \mathcal{O}(k^2 P + k^5)$.
\end{proof}

Due to~\cref{thm:algorithm}, every configuration can be efficiently transformed into, e.g., a $3$-sponge of comparable size.
In particular, given some configuration $C_1$ with bounding box perimeter $P_1$, we can obtain a $3$-sponge of bounding box perimeter at most $\nicefrac{3}{2}P_1$.
Combining this with~\cref{prop:k-sponge-reconfiguration}, we obtain the following:

\begin{corollary}
    \label{thm:algorithm-labeled}
    For any instance $\mathcal{I}$ of {\textsc{Labeled Parallel Sliding Squares}}, we can compute a schedule of $\mathcal{O}(P_1+P_2)$ transformations in polynomial time.
\end{corollary}

Note, however, that this process is explicitly not guaranteed to yield an in-place (or even weakly in-place) schedule, as building a sponge from a compact configuration may exceed the original bounding box by a linear amount in either dimension.

If we restrict ourselves to weakly in-place reconfiguration, a worst-case input for our approach is then a pair of configurations that share a common, rectangular silhouette.
We can handle these using a slightly slower variant of the approach outlined above, yielding the following statement for in-place reconfiguration.
\begin{theorem}
    \label{thm:algorithm-labeled-in-place}
    For any instance $\mathcal{I}$ of {\textsc{Labeled Parallel Sliding Squares}}, we can compute a weakly in-place schedule of $\mathcal{O}(P_1^2+P_2^2)$ transformations in polynomial time.
\end{theorem}
\begin{proof}
    Due to~\cref{thm:algorithm}, we may assume without loss of generality, that the two configurations have a common silhouette~$C$ which is a solid rectangle of width $w$ and height $h\leq w$.
    It then suffices to argue that we can arbitrarily swap modules efficiently.

    The routing number of this silhouette's dual graph is linear in its width and height due to~\cref{thm:alon-chung-graham}, i.e., $\mathsf{rt}(\operatorname{dual}(C))\in\mathcal{O}(w+h)$.
    To perform a single routing operation, we divide the configuration into vertical slices of width $\mathcal{O}(1)$.
    Each slice then contains at most $\mathcal{O}(h)$ many pairs of adjacent modules that need to be swapped as part of the routing operation.
    In each slice, we create two empty cells by moving two modules out of the bounding box at either the top or the bottom of the configuration.
    Using simple schedules such as the one outlined in~\cref{fig:sponge-adjacent-swap}, we can perform all of these swaps in constant time each, i.e., no more than $\mathcal{O}(h)$ per routing operation.

    The total makespan is then bounded in $\mathcal{O}((w+h)h)\subseteq\mathcal{O}(P^2)$.
\end{proof}
%
\section{Conclusions and future work}
\label{sec:conclusions}
We have presented several new results for reconfiguration in the sliding squares model, making full use of parallel motion, and distinguished the two settings of labeled and unlabeled modules.
In particular, we showed that deciding whether there exists a single transformation that reconfigures one unlabeled configuration into another is already \mbox{\NP-complete}.
Additionally, we provided algorithmic results for the unlabeled~variant, including a polynomial-time algorithm that performs in-place reconfiguration in $\mathcal{O}(P)$ transformations, where $P$ denotes the perimeter of the union of the configurations' bounding boxes.

In addition, we adapted our algorithm to the labeled setting by ``sorting'' the modules in the $xy$-monotone configuration using~$\mathcal{O}(P)$ transformations.
However, this requires a relaxation of the in-place requirement.
Furthermore, deciding whether a schedule of makespan~$1$ exists in the labeled setting is straightforward, while the problem becomes \NP-complete for a makespan of $2$.

Our algorithmic results are only worst-case optimal, and there are a number of possible generalizations and extensions of the setting.
Previous work in the sequential setting has progressed from two dimensions to three.
Can our approach be extended to higher dimensions?
We are optimistic that significant speedup can be achieved, but the intricacies of three-dimensional topology may require additional tools.

Another aspect is a refinement of the involved parameters, along with
fixed parameter tractability.
Our hardness proofs show that the investigated problems are not \FPT\ when parameterized by the number of transformations in the output.
The number of modules in the input configuration is not a suitable parameter since it is the input size.
Are there other parameters for which these problems are \FPT?
An interesting candidate would be the size of the symmetric difference between the two input configurations.
In our hardness proofs, the size of the symmetric difference is linear in the size of the \SAT\ instance.

    \bibliographystyle{plainurl}%
    \bibliography{bibliography}

\begin{thebibliography}{10}

\bibitem{abel.akitaya.kominers.ea2024universal-in-place}
Zachary Abel, Hugo~A. Akitaya, Scott~Duke Kominers, Matias Korman, and
  Frederick Stock.
\newblock A universal in-place reconfiguration algorithm for sliding
  cube-shaped robots in a quadratic number of moves.
\newblock In {\em Symposium on Computational Geometry (SoCG)}, pages 1:1--1:14,
  2024.
\newblock \href {https://doi.org/10.4230/LIPIcs.SoCG.2024.1}
  {\path{doi:10.4230/LIPIcs.SoCG.2024.1}}.

\bibitem{akitaya.arkin.damian.ea2021universal-reconfiguration}
Hugo~A. Akitaya, Esther~M. Arkin, Mirela Damian, Erik~D. Demaine, Vida
  Dujmovic, Robin~Y. Flatland, Matias Korman, Bel{\'{e}}n Palop, Irene Parada,
  Andr{\'{e}} van Renssen, and Vera Sacrist{\'{a}}n.
\newblock Universal reconfiguration of facet-connected modular robots by
  pivots: The $\mathcal{O}(1)$ musketeers.
\newblock {\em Algorithmica}, 83(5):1316--1351, 2021.
\newblock \href {https://doi.org/10.1007/s00453-020-00784-6}
  {\path{doi:10.1007/s00453-020-00784-6}}.

\bibitem{akitaya.demaine.gonczi.ea2021characterizing-universal}
Hugo~A. Akitaya, Erik~D. Demaine, Andrei Gonczi, Dylan~H. Hendrickson, Adam
  Hesterberg, Matias Korman, Oliver Korten, Jayson Lynch, Irene Parada, and
  Vera Sacrist\'{a}n.
\newblock Characterizing universal reconfigurability of modular pivoting
  robots.
\newblock In {\em Symposium on Computational Geometry (SoCG)}, pages
  10:1--10:20, 2021.
\newblock \href {https://doi.org/10.4230/LIPIcs.SoCG.2021.10}
  {\path{doi:10.4230/LIPIcs.SoCG.2021.10}}.

\bibitem{akitaya.demaine.korman.ea2022compacting-squares}
Hugo~A. Akitaya, Erik~D. Demaine, Matias Korman, Irina Kostitsyna, Irene
  Parada, Willem Sonke, Bettina Speckmann, Ryuhei Uehara, and Jules Wulms.
\newblock Compacting squares: Input-sensitive in-place reconfiguration of
  sliding squares.
\newblock In {\em Scandinavian Symposium and Workshops on Algorithm Theory
  (SWAT)}, pages 4:1--4:19, 2022.
\newblock \href {https://doi.org/10.4230/LIPIcs.SWAT.2022.4}
  {\path{doi:10.4230/LIPIcs.SWAT.2022.4}}.

\bibitem{alon.chung.graham1994routing}
Noga Alon, Fan R.~K. Chung, and Ronald~L. Graham.
\newblock Routing permutations on graphs via matchings.
\newblock {\em {SIAM} Journal on Discrete Mathematics}, 7(3):513--530, 1994.
\newblock \href {https://doi.org/10.1137/S0895480192236628}
  {\path{doi:10.1137/S0895480192236628}}.

\bibitem{aloupis.collette.damian.ea2009linear-reconfiguration}
Greg Aloupis, S{\'e}bastien Collette, Mirela Damian, Erik~D. Demaine, Robin
  Flatland, Stefan Langerman, Joseph O'Rourke, Suneeta Ramaswami, Vera
  Sacrist{\'a}n, and Stefanie Wuhrer.
\newblock Linear reconfiguration of cube-style modular robots.
\newblock {\em Computational Geometry}, 42(6-7):652--663, 2009.
\newblock \href {https://doi.org/10.1016/J.COMGEO.2008.11.003}
  {\path{doi:10.1016/J.COMGEO.2008.11.003}}.

\bibitem{becker.fekete.keldenich.ea2018coordinated-motion}
Aaron~T. Becker, S{\'a}ndor~P. Fekete, Phillip Keldenich, Matthias Konitzny,
  Lillian Lin, and Christian Scheffer.
\newblock Coordinated motion planning: The video.
\newblock In {\em Symposium on Computational Geometry (SoCG)}, pages
  74:1--74:6, 2018.
\newblock \href {https://doi.org/10.4230/LIPICS.SOCG.2018.74}
  {\path{doi:10.4230/LIPICS.SOCG.2018.74}}.

\bibitem{bourgeois.fekete.kosfeld.ea2022space-ants}
Julien Bourgeois, S\'{a}ndor~P. Fekete, Ramin Kosfeld, Peter Kramer, Beno\^{i}t
  Piranda, Christian Rieck, and Christian Scheffer.
\newblock Space ants: Episode {II} - {C}oordinating connected catoms.
\newblock In {\em Symposium on Computational Geometry (SoCG)}, pages
  65:1--65:6, 2022.
\newblock \href {https://doi.org/10.4230/LIPIcs.SoCG.2022.65}
  {\path{doi:10.4230/LIPIcs.SoCG.2022.65}}.

\bibitem{butler.rus2003distributed-planning}
Zack Butler and Daniela Rus.
\newblock Distributed planning and control for modular robots with
  unit-compressible modules.
\newblock {\em The International Journal of Robotics Research}, 22(9):699--715,
  2003.
\newblock \href {https://doi.org/10.1177/02783649030229002}
  {\path{doi:10.1177/02783649030229002}}.

\bibitem{connor.michail.skretas2024all-for}
Matthew Connor, Othon Michail, and George Skretas.
\newblock All for one and one for all: An $\mathcal{O}(1)$-musketeers universal
  transformation for rotating robots.
\newblock In {\em Symposium on Algorithmic Foundations of Dynamic Networks
  (SAND)}, pages 9:1--9:20, 2024.
\newblock \href {https://doi.org/10.4230/LIPIcs.SAND.2024.9}
  {\path{doi:10.4230/LIPIcs.SAND.2024.9}}.

\bibitem{berg.khosravi2012optimal-binary}
Mark de~Berg and Amirali Khosravi.
\newblock Optimal binary space partitions for segments in the plane.
\newblock {\em International Journal on Computational Geometry and
  Applications}, 22(3):187--206, 2012.
\newblock \href {https://doi.org/10.1142/S0218195912500045}
  {\path{doi:10.1142/S0218195912500045}}.

\bibitem{demaine.fekete.keldenich.ea2019coordinated-motion}
Erik~D. Demaine, S{\'a}ndor~P. Fekete, Phillip Keldenich, Henk Meijer, and
  Christian Scheffer.
\newblock Coordinated motion planning: Reconfiguring a swarm of labeled robots
  with bounded stretch.
\newblock {\em SIAM Journal on Computing}, 48(6):1727--1762, 2019.
\newblock \href {https://doi.org/10.1137/18M1194341}
  {\path{doi:10.1137/18M1194341}}.

\bibitem{dietzfelbinger.karlin.ea1994dynamic-perfect}
Martin Dietzfelbinger, Anna Karlin, Kurt Mehlhorn, Friedhelm Meyer auf~der
  Heide, Hans Rohnert, and Robert~E. Tarjan.
\newblock Dynamic perfect hashing: Upper and lower bounds.
\newblock {\em SIAM Journal on Computing}, 23(4):738--761, 1994.
\newblock \href {https://doi.org/10.1137/S0097539791194094}
  {\path{doi:10.1137/S0097539791194094}}.

\bibitem{dumitrescu.pach2006pushing-squares}
Adrian Dumitrescu and J{\'{a}}nos Pach.
\newblock Pushing squares around.
\newblock {\em Graphs and Combinatorics}, 22(1):37--50, 2006.
\newblock \href {https://doi.org/10.1007/s00373-005-0640-1}
  {\path{doi:10.1007/s00373-005-0640-1}}.

\bibitem{dumitrescu.suzuki.yamashita2004formations-for}
Adrian Dumitrescu, Ichiro Suzuki, and Masafumi Yamashita.
\newblock Formations for fast locomotion of metamorphic robotic systems.
\newblock {\em International Journal of Robotics Research}, 23(6):583--593,
  2004.
\newblock \href {https://doi.org/10.1177/0278364904039652}
  {\path{doi:10.1177/0278364904039652}}.

\bibitem{dumitrescu.suzuki.yamashita2004motion-planning}
Adrian Dumitrescu, Ichiro Suzuki, and Masafumi Yamashita.
\newblock Motion planning for metamorphic systems: feasibility, decidability,
  and distributed reconfiguration.
\newblock {\em Transactions on Robotics}, 20(3):409--418, 2004.
\newblock \href {https://doi.org/10.1109/TRA.2004.824936}
  {\path{doi:10.1109/TRA.2004.824936}}.

\bibitem{fekete.keldenich.kosfeld.ea2023connected-coordinated}
S{\'{a}}ndor~P. Fekete, Phillip Keldenich, Ramin Kosfeld, Christian Rieck, and
  Christian Scheffer.
\newblock Connected coordinated motion planning with bounded stretch.
\newblock {\em Autonomous Agents and Multi-Agent Systems}, 37(2):43, 2023.
\newblock \href {https://doi.org/10.1007/S10458-023-09626-5}
  {\path{doi:10.1007/S10458-023-09626-5}}.

\bibitem{fekete.kosfeld.kramer.ea2024coordinated-motion}
S{\'{a}}ndor~P. Fekete, Ramin Kosfeld, Peter Kramer, Jonas Neutzner, Christian
  Rieck, and Christian Scheffer.
\newblock Coordinated motion planning: Multi-agent path finding in a densely
  packed, bounded domain.
\newblock In {\em International Symposium on Algorithms and Computation
  (ISAAC)}, pages 29:1--29:15, 2024.
\newblock \href {https://doi.org/10.4230/LIPIcs.ISAAC.2024.29}
  {\path{doi:10.4230/LIPIcs.ISAAC.2024.29}}.

\bibitem{fekete.kramer.rieck.ea2024efficiently-reconfiguring}
S{\'{a}}ndor~P. Fekete, Peter Kramer, Christian Rieck, Christian Scheffer, and
  Arne Schmidt.
\newblock Efficiently reconfiguring a connected swarm of labeled robots.
\newblock {\em Autonomous Agents and Multi-Agent Systems}, 38(2):39, 2024.
\newblock \href {https://doi.org/10.1007/S10458-024-09668-3}
  {\path{doi:10.1007/S10458-024-09668-3}}.

\bibitem{fitch.butler.rus2003reconfiguration-planning}
Robert Fitch, Zack~J. Butler, and Daniela Rus.
\newblock Reconfiguration planning for heterogeneous self-reconfiguring robots.
\newblock In {\em International Conference on Intelligent Robots and Systems
  (IROS)}, pages 2460--2467, 2003.
\newblock \href {https://doi.org/10.1109/IROS.2003.1249239}
  {\path{doi:10.1109/IROS.2003.1249239}}.

\bibitem{fukuda.nakagawa.kawauchi.ea1989structure-decision}
Toshio Fukuda, Seiya Nakagawa, Yoshio Kawauchi, and Martin Buss.
\newblock Structure decision method for self organising robots based on cell
  structures-cebot.
\newblock In {\em International Conference on Robotics and Automation (ICRA)},
  pages 695--696, 1989.
\newblock \href {https://doi.org/10.1109/ROBOT.1989.100066}
  {\path{doi:10.1109/ROBOT.1989.100066}}.

\bibitem{hurtado.molina.ramaswami.ea2015distributed-reconfiguration}
Ferran Hurtado, Enrique Molina, Suneeta Ramaswami, and Vera~Sacrist{\'{a}}n
  Adinolfi.
\newblock Distributed reconfiguration of 2d lattice-based modular robotic
  systems.
\newblock {\em Autonomous Robots}, 38(4):383--413, 2015.
\newblock \href {https://doi.org/10.1007/S10514-015-9421-8}
  {\path{doi:10.1007/S10514-015-9421-8}}.

\bibitem{kostitsyna.ophelders.parada.ea2024optimal-in-place}
Irina Kostitsyna, Tim Ophelders, Irene Parada, Tom Peters, Willem Sonke, and
  Bettina Speckmann.
\newblock Optimal in-place compaction of sliding cubes.
\newblock In {\em Scandinavian Symposium and Workshops on Algorithm Theory
  (SWAT)}, pages 31:1--31:14, 2024.
\newblock \href {https://doi.org/10.4230/LIPICS.SWAT.2024.31}
  {\path{doi:10.4230/LIPICS.SWAT.2024.31}}.

\bibitem{kotay.rus2000algorithms-for}
Keith~D. Kotay and Daniela~L. Rus.
\newblock Algorithms for self-reconfiguring molecule motion planning.
\newblock In {\em International Conference on Intelligent Robots and Systems
  (IROS)}, pages 2184--2193, 2000.
\newblock \href {https://doi.org/10.1109/IROS.2000.895294}
  {\path{doi:10.1109/IROS.2000.895294}}.

\bibitem{michail.skretas.spirakis2019on-transformation}
Othon Michail, George Skretas, and Paul~G. Spirakis.
\newblock On the transformation capability of feasible mechanisms for
  programmable matter.
\newblock {\em Journal of Computer and System Sciences}, 102:18--39, 2019.
\newblock \href {https://doi.org/10.1016/j.jcss.2018.12.001}
  {\path{doi:10.1016/j.jcss.2018.12.001}}.

\bibitem{moreno.sacristan2020reconfiguring-sliding}
Joel Moreno and Vera Sacrist{\'a}n.
\newblock Reconfiguring sliding squares in-place by flooding.
\newblock In {\em European Workshop on Computational Geometry (EuroCG)}, pages
  32:1--32:7, 2020.
\newblock URL:
  \url{https://computational-geometry.org/abstracts/eurocg/2020.pdf}.

\bibitem{nguyen.guibas.yim2001controlled-module}
An~Nguyen, Leonidas~J. Guibas, and Mark Yim.
\newblock Controlled module density helps reconfiguration planning.
\newblock {\em New Directions in Algorithmic and Computational Robotics}, pages
  23--36, 2001.

\bibitem{parada.sacristan.silveira2021new-meta-module}
Irene Parada, Vera Sacrist{\'{a}}n, and Rodrigo~I. Silveira.
\newblock A new meta-module design for efficient reconfiguration of modular
  robots.
\newblock {\em Autonomous Robots}, 45(4):457--472, 2021.
\newblock \href {https://doi.org/10.1007/S10514-021-09977-6}
  {\path{doi:10.1007/S10514-021-09977-6}}.

\bibitem{rus.vona2001crystalline-robots}
Daniela Rus and Marsette Vona.
\newblock Crystalline robots: Self-reconfiguration with compressible unit
  modules.
\newblock {\em Autonomous Robots}, 10:107--124, 2001.
\newblock \href {https://doi.org/10.1023/A:1026504804984}
  {\path{doi:10.1023/A:1026504804984}}.

\bibitem{schwartz.sharir1983on-piano}
Jacob~T. Schwartz and Micha Sharir.
\newblock On the piano movers' problem: {III.} {C}oordinating the motion of
  several independent bodies: the special case of circular bodies moving amidst
  polygonal barriers.
\newblock {\em International Journal of Robotics Research}, 2(3):46--75, 1983.
\newblock \href {https://doi.org/10.1177/027836498300200304}
  {\path{doi:10.1177/027836498300200304}}.

\bibitem{vassilvitskii.yim.suh2002complete-local}
Serguei Vassilvitskii, Mark Yim, and John Suh.
\newblock A complete, local and parallel reconfiguration algorithm for cube
  style modular robots.
\newblock In {\em International Conference on Robotics and Automation
  {(ICRA)}}, pages 117--122, 2002.
\newblock \href {https://doi.org/10.1109/ROBOT.2002.1013348}
  {\path{doi:10.1109/ROBOT.2002.1013348}}.

\bibitem{wolters2024parallel-algorithms}
Matthijs Wolters.
\newblock Parallel algorithms for sliding squares.
\newblock Master's thesis, Utrecht University, 2024.

\bibitem{yim.shen.salemi.ea2007modular-self-reconfigurable}
Mark Yim, Wei{-}Min Shen, Behnam Salemi, Daniela Rus, Mark Moll, Hod Lipson,
  Eric Klavins, and Gregory~S. Chirikjian.
\newblock Modular self-reconfigurable robot systems [grand challenges of
  robotics].
\newblock {\em Robotics and Automation Magazine}, 14(1):43--52, 2007.
\newblock \href {https://doi.org/10.1109/MRA.2007.339623}
  {\path{doi:10.1109/MRA.2007.339623}}.

\end{thebibliography}
\end{document}